\documentclass[11pt,a4paper]{article}
\pdfoutput=1
\usepackage{jheppub}
\usepackage[T1]{fontenc}
\usepackage{lmodern}
\usepackage{amsmath,amssymb,mathtools,amsthm}
\usepackage{bm}
\usepackage{tikz}
\usepackage{booktabs}

\newcommand{\cC}{\mathcal C}
\newcommand{\cD}{\mathcal D}
\newcommand{\Li}{\operatorname{Li}}
\newcommand{\WZW}{\mathrm{WZW}}
\newcommand{\ii}{\mathrm i}
\newcommand{\Spin}{\operatorname{Spin}}
\newcommand{\dd}{\,\mathrm d}

\newtheorem{theorem}{Theorem}
\newtheorem{proposition}[theorem]{Proposition}
\newtheorem{lemma}[theorem]{Lemma}
\newtheorem{corollary}[theorem]{Corollary}

\title{Maximal Total Quantum Dimension at Bounded Rank in WZW Modular Tensor Categories}

\author[a]{Ce Shen}
\affiliation[a]{Beijing Institute of Mathematical Sciences and Applications (BIMSA),
Beijing, China}
\emailAdd{shence@bimsa.cn}

\abstract{
We determine the maximal torus topological entanglement entropy (TEE)
at bounded torus ground-state degeneracy within
simply connected untwisted Wess--Zumino--Witten theories.  For a fixed modular
tensor category $\mathcal C$, maximization over normalized torus ground states
gives $\max_\psi\Gamma_{T^2}(\psi)=2\log\mathcal D(\mathcal C)$, attained in
particular by the vacuum-flux state.  The remaining problem is therefore to
maximize total quantum dimension at bounded categorical rank.  For a simple
WZW category $\mathcal C(\mathfrak g,k)$, write $r(\mathfrak g,k)$ and
$\mathcal D(\mathfrak g,k)$ for its categorical rank and total quantum
dimension, and define
$F_{\WZW}(R)=\sup_{r(\mathfrak g,k)\leq R}2\log\mathcal D(\mathfrak g,k)$.
We prove the sharp asymptotic law
$\lim_{R\to\infty}\frac{F_{\WZW}(R)}{(\log_2R)^2}
=\frac{7\zeta(3)}{4\pi^2}$.
The balanced symplectic sequence $Sp(2n)_n$ attains the coefficient.  The
constant arises from the odd Fourier modes of the type-$C$ root product at
equal rank and level.  A sharp entropy--spectral inequality, rank-level
duality, and fixed-rank estimates give the global upper bound.  A separate
corollary extends the same leading law to semisimple WZW categories,
equivalently finite Deligne products of simple factors.  The unrestricted
TQFT envelope remains open; the theorem is sharp within the stated WZW class
and supplies a constructive lower bound for the general problem.}

\keywords{{\small Chern--Simons Theories, Topological Entanglement Entropy,
Topological Field Theories}}

\begin{document}

\maketitle

\section{Introduction}

% \paragraph{Physical question.}
In $(2+1)$ dimensions, topological order is encoded by long-range
entanglement rather than a local order parameter
\cite{Wen1990,ChenGuWen,WenBosonic,WenZoo}.  Its universal data determine
fractional and non-Abelian quasiparticles, topology-dependent ground-state
degeneracy (GSD), and protected edge structure
\cite{WenNiu,Wen1995,NayakEtAl}.  Anyon fusion and braiding describe the bulk,
chiral operator algebras describe quantum Hall edges, and modular
transformations act on the ground-state space
\cite{WenWu,WenWuHatsugai,Witten}.

The torus GSD and topological entanglement entropy probe complementary aspects
of this order.  For a phase described by a unitary modular tensor category
$\cC$, the torus GSD is the categorical rank
$r(\cC)\equiv\operatorname{rk}(\cC)
=\lvert\operatorname{Irr}(\cC)\rvert$, the number of isomorphism classes of
simple superselection sectors.  The intrinsic total quantum dimension
$\cD(\cC)$ instead weights those sectors by their quantum dimensions and
measures their collective fusion complexity.
For a disk in the vacuum sector, the topological entanglement
entropy in the positive-magnitude convention adopted below
\cite{KitaevPreskill,LevinWen} is $\log\cD$, the logarithm of the total quantum
dimension, as explained in Section~\ref{sec:setup}.  Noncontractible
cuts distinguish flux sectors and can reveal further modular information
\cite{DongEtAl,ZhangEtAl,LuoEtAl}.

For the torus bipartition studied below, the cut has two entangling circles,
and its TEE $\Gamma_{T^2}(\psi)$ depends on the normalized ground state
$\psi\in\mathcal H(T^2)$.  Here $\mathcal H(T^2)$ is the Hilbert space
assigned to the torus.  The extremal problem therefore has two
stages.  First, for a fixed topological phase, which torus ground state
maximizes $\Gamma_{T^2}$?  Second, after optimizing over the state, how large
can this TEE be when the torus GSD is at most $R$?

Section~\ref{sec:setup} performs the first optimization.  It proves
\(
%  \max_{\substack{\psi\in\mathcal H(T^2)\\\lVert\psi\rVert=1}}
\max_{\substack{\psi\in\mathcal H(T^2)}}
 \Gamma_{T^2}(\psi)=2\log\cD.
\)
The maximizers are exactly the definite Abelian-flux states.  Since the
vacuum has quantum dimension one, the vacuum-flux state is always a maximizer
and provides the canonical representative used later.  Only after this
state-space reduction does the remaining problem become an optimization of
$2\log\cD$ over topological phases.
% We write $q_R=\log_2R$ only as a
% convenient logarithmic measure of the cutoff; it is not assumed to be an
% operational quantum-computing capacity.

% \paragraph{Invariant formulation.}
The state-optimized problem has a useful closed-manifold formulation.  Let
$\mathcal T$ be a unitary bosonic semisimple three-dimensional TQFT and let
$\cC$ denote its
modular category of line operators.\footnote{The qualifier ``bosonic'' is important.  A fermionic topological phase with a
transparent physical fermion is naturally described by a super-modular,
rather than modular, category \cite{GuWangWen,BruillardFermionic}.  Such phases require an additional choice of
modular extension and are not included in the optimization below.  Ordinary
bosonic WZW Chern--Simons theories give modular categories and fit directly
into this framework.}
Let $\mathcal H_{\mathcal T}(T^2)$ be the
Hilbert space assigned to the torus, and let $S_{00}$ be the vacuum--vacuum
entry of the modular $S$ matrix.  In the standard normalization and canonical
framing, gluing and the state--operator correspondence give \cite{Witten,Turaev}
\begin{equation}
 Z_{\mathcal T}(T^3)=\dim\mathcal H_{\mathcal T}(T^2)
 =r(\cC),
 \qquad
 Z_{\mathcal T}(S^3)=S_{00}=\frac1{\cD(\cC)}.
 \label{eq:tqft-partitions}
\end{equation}
Thus $Z(T^3)$ records the torus GSD, while the inverse vacuum amplitude on
$S^3$ records $\cD$.  The unrestricted invariant version of the physical
question is the cutoff envelope of the state-optimized TEE:
\begin{align}
 F(R)
 &:={}
 \sup_{\substack{\mathcal T\ \mathrm{unitary\ bosonic\ semisimple}\\
                  Z_{\mathcal T}(T^3)\leq R}}
 \quad\max_{\substack{\psi\in\mathcal H_{\mathcal T}(T^2)}}
 \Gamma_{T^2}^{\mathcal T}(\psi)\nonumber\\
 &=\sup_{\substack{\mathcal T\ \mathrm{unitary\ bosonic\ semisimple}\\
                   Z_{\mathcal T}(T^3)\leq R}}
 \bigl[-2\log Z_{\mathcal T}(S^3)\bigr]
 =\sup_{\substack{\cC\ \mathrm{unitary\ modular}\\
                   r(\cC)\leq R}}
 2\log\cD(\cC).
 \label{eq:general-envelope}
\end{align}
Thus $\Gamma_{T^2}^{\mathcal T}(\psi)$ is the TEE of one torus state in one
phase, the inner maximum selects the largest TEE available in that phase, and
$F(R)$ takes the supremum of these optimized values over all phases whose
torus GSD is at most $R$.  The remaining equalities use
Eq.~\eqref{eq:state-maximization} and record the modular-category realization
of the normalized invariants.  The total quantum dimension should not be
confused with a closed-surface Hilbert-space dimension.

\paragraph{Simple WZW envelope.}
This paper does not determine $F(R)$.  To obtain an exact result while
retaining a broad and physically important non-Abelian family, we turn to
WZW--Chern--Simons theories.  Their bulk Wilson lines, torus ground states, and
chiral boundary sectors are governed by the same affine-Lie-algebra modular
data \cite{Witten,NayakEtAl}.  The numbers of sectors and their quantum
dimensions are known exactly, making this class a controlled setting in which
the competition between GSD and entanglement can be solved asymptotically.

For a finite-dimensional simple Lie algebra $\mathfrak g$ and positive
integer level $k$, let $\cC(\mathfrak g,k)$ denote the simply connected WZW
modular category of the untwisted affine algebra $\widehat{\mathfrak g}_k$.
We denote its sector count and total quantum dimension by
$r(\mathfrak g,k):=r(\cC(\mathfrak g,k))$ and
$\cD(\mathfrak g,k):=\cD(\cC(\mathfrak g,k))$, respectively.
More precisely, let $\mathcal H_{\mathfrak g,k}(T^2)$ be its torus ground-state Hilbert space.
We study the restricted envelope
\begin{align}
 F_{\WZW}(R)
 &:={}
 \sup_{\substack{\mathfrak g\ \mathrm{simple},\ k\in\mathbb Z_{>0}\\
                  r(\mathfrak g,k)\leq R}}
 \max_{\substack{\psi\in\mathcal H_{\mathfrak g,k}(T^2)}}
 \Gamma_{T^2}^{(\mathfrak g,k)}(\psi)\nonumber\\
 &=\sup_{\substack{\mathfrak g\ \mathrm{simple},\ k\in\mathbb Z_{>0}\\
                  r(\mathfrak g,k)\leq R}}
 2\log\cD(\mathfrak g,k),
 \label{eq:wzw-envelope}
\end{align}
where the first line is the two-stage physical optimization and the second
uses Eq.~\eqref{eq:state-maximization}.  Thus $F_{\WZW}(R)$ is precisely the
maximal torus TEE in the simple WZW class at cutoff $R$.
Our principal result is the sharp asymptotic law
\begin{equation}
 \lim_{R\to\infty}\frac{F_{\WZW}(R)}{(\log_2R)^2}
 =\frac{7\zeta(3)}{4\pi^2}.
 \label{eq:intro-mainlaw}
\end{equation}
The balanced symplectic sequence $Sp(2n)_n$ with $n\to\infty$ attains this
coefficient.
% Balanced orthogonal sequences have the same leading
% coefficient, while type $A$ attains half of it.

\paragraph{Semisimple extension.}
The principal theorem concerns one simple WZW factor at a time.  To formulate
the separate finite-stacking extension, choose any $m\in\mathbb Z_{>0}$,
simple Lie algebras $\mathfrak g_1,\ldots,\mathfrak g_m$, and independent
positive integral levels $k_1,\ldots,k_m$.  For the semisimple Lie algebra,
\begin{equation}
 \mathfrak g_{\mathrm{ss}}=\bigoplus_{i=1}^m\mathfrak g_i,
 \qquad
 \cC(\mathfrak g_{\mathrm{ss}},\boldsymbol k)
 =\boxtimes_{i=1}^m\cC(\mathfrak g_i,k_i).
 \label{eq:semisimple-factorization}
\end{equation}
The second identity means the semisimple theory is the product of the simply connected untwisted affine theory \cite{KacPeterson,Turaev}.  We define
the admissible stacked class by
\begin{equation}
 \mathfrak W^{\boxtimes}
 :=\bigcup_{m\geq1}
 \left\{
  \boxtimes_{i=1}^m\cC(\mathfrak g_i,k_i):
  \mathfrak g_i\ \mathrm{simple},\ k_i\in\mathbb Z_{>0}
 \right\}.
 \label{eq:stacked-class}
\end{equation}
Because every finite-dimensional semisimple Lie algebra is a finite direct
sum of simple components, $\mathfrak W^{\boxtimes}$ is precisely the class of
simply connected untwisted semisimple WZW categories with independent levels.
The union is over all positive integers $m$.  At a
given cutoff, the supremum therefore ranges over every finite product whose
categorical rank is at most $R$.  Its envelope is
\begin{equation}
 F_{\WZW}^{\boxtimes}(R)=
 \sup_{\substack{\cC\in\mathfrak W^{\boxtimes}\\r(\cC)\leq R}}
 2\log\cD(\cC).
 \label{eq:stacked-envelope}
\end{equation}
A separate corollary proves
\begin{equation}
 \lim_{R\to\infty}\frac{F_{\WZW}^{\boxtimes}(R)}{(\log_2R)^2}
 =\frac{7\zeta(3)}{4\pi^2},
 \label{eq:intro-stacked-law}
\end{equation}
so finite stacking does not improve the leading coefficient.

\paragraph{Scope and proof strategy.}
For the chosen torus bipartition, $F_{\WZW}$ and
$F_{\WZW}^{\boxtimes}$ are the state-optimized torus TEEs in the
simple and semisimple WZW classes, respectively; the vacuum flux provides a
canonical maximizing state in every category.  The results give an explicit
lower bound for the unrestricted problem and sharp extremal theorems for
these Lie-theoretic WZW--Chern--Simons families.  Whether larger envelopes arise
from other modular categories, cosets, orbifolds, Dijkgraaf--Witten theories
\cite{DijkgraafWitten}, gauged theories, or non-simply-connected
Chern--Simons theories remains open.

The proof separates the saturating construction from the global upper bound.
For the upper bound it treats explicitly the four possible parameter regimes:
rank and level growing at fixed ratio, both growing at an unbalanced ratio,
fixed Lie rank with growing level, and fixed level with growing Lie rank.
It combines proportional root-density asymptotics, an exact algebraic proof
of the sharp entropy--spectral inequality, vacuum-entry rank-level identities,
and an explicit passage from sequential bounds to the cutoff supremum.
Technical WZW data and elementary rational bounds used in the proof are
collected in the appendices.

\section{Topological Setup}
\label{sec:setup}

We now perform the first stage of the extremal problem: for a fixed
topological phase, maximize the torus TEE over its entire
ground-state Hilbert space.  This step is what reduces the original
state-dependent TEE problem to the categorical envelope studied in the rest
of the paper.

A unitary modular tensor category $\cC$ has a modular $S$ matrix and simple
line operators indexed by $a\in\operatorname{Irr}(\cC)$.  Their positive
quantum dimensions are $d_a=S_{0a}/S_{00}$; the tensor unit is $0$ and has
$d_0=1$
\cite{NayakEtAl,WenBosonic,GuWangWen}.  Its total quantum dimension is the
intrinsic invariant
\begin{equation}
 \cD(\cC)^2=\sum_{a\in\operatorname{Irr}(\cC)}d_a^2,
 \qquad S_{00}=\frac1{\cD(\cC)}>0.
 \label{eq:total-dimension}
\end{equation}
Accordingly $Z(S^3)=1/\cD$ in the normalization of
Eq.~\eqref{eq:tqft-partitions}.  This weighted fusion invariant should not be
confused with a closed-surface Hilbert-space dimension.

Quantizing on a torus gives one basis state for each simple line, viewed as its
flux through a noncontractible cycle.  Hence
\begin{equation}
 r(\cC)\equiv\operatorname{rk}(\cC)
 =Z(T^3)=\dim\mathcal H(T^2)
 =\lvert\operatorname{Irr}(\cC)\rvert,
 \qquad q(\cC)=\log_2r(\cC).
 \label{eq:logarithmic-rank}
\end{equation}
Here $\mathcal H(T^2)$ denotes the torus ground-state Hilbert space.
We use $q$ only as the binary logarithmic size of the torus Hilbert space; it
is not an operational quantum-computing capacity.  No claim is made about
braiding universality, encoding, controllability, error correction, or
fault-tolerant gates.  Below, $r(\mathfrak g,k)$ is the exact WZW sector count,
$R$ is the external cutoff, and $q_R=\log_2R$ is only a convenient change of
units.  Throughout, the generic Lie-algebra rank is denoted by $\ell$.  In
the classical labels $A_{n-1},B_n,C_n,D_n$, the symbol $n$ is only the
family-size parameter: $\ell=n-1$ for $A_{n-1}$ and $\ell=n$ for
$B_n,C_n,D_n$.

\begin{figure}[t]
  \centering
  \includegraphics[width=0.8\linewidth]{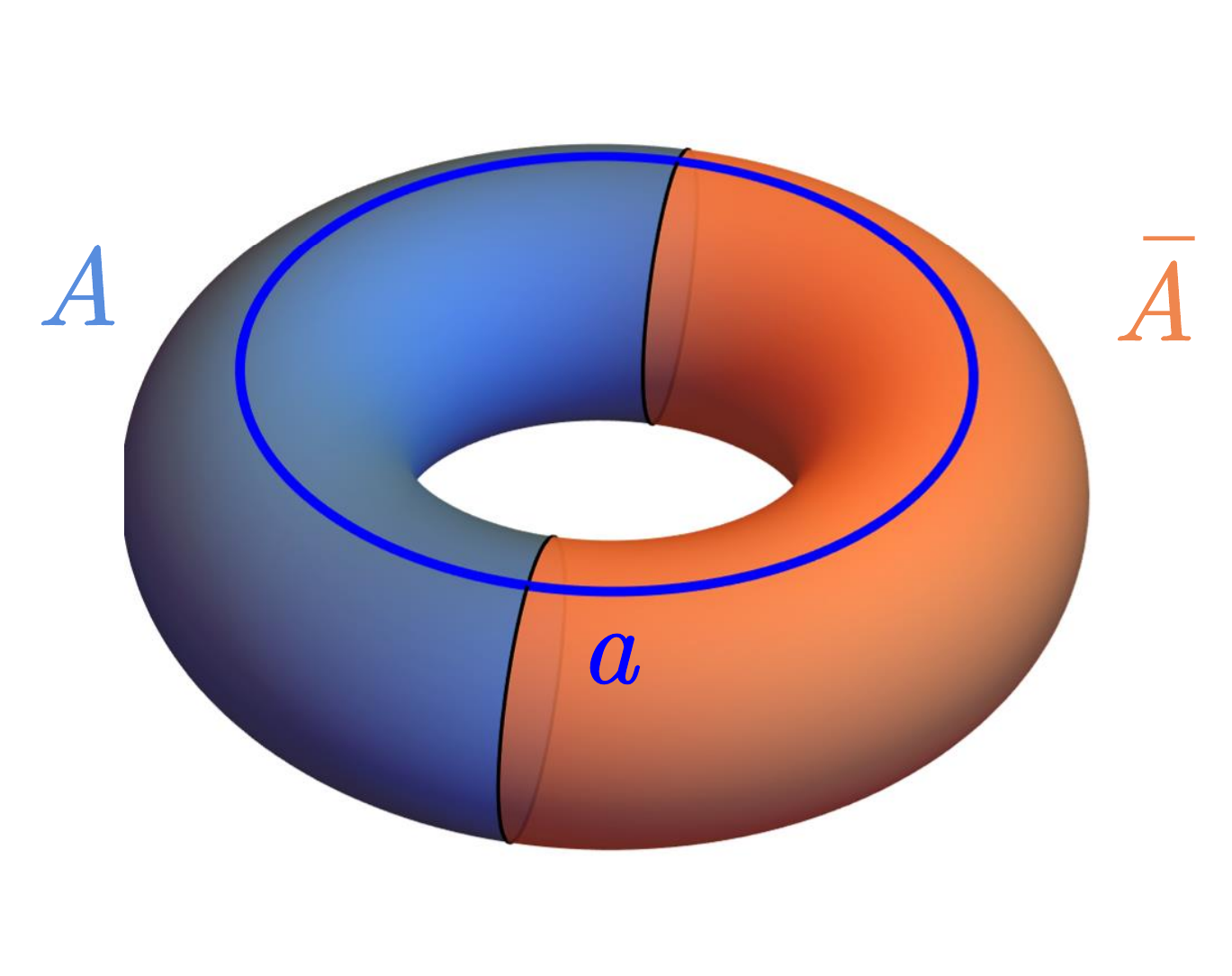}
  \caption{A torus divided into cylinders $A$ and $\bar A$.  Their two
  circular boundaries, of lengths $L_1$ and $L_2$, form the
  entanglement cut.  The blue loop carries anyon flux $a$.}
  \label{fig:torus}
\end{figure}

Cut the torus into two cylinders $A$ and $\bar A$, as in
Fig.~\ref{fig:torus}.  Let $\lvert a\rangle$ be the ground state with anyon
flux $a$ through the cylinder, and write
$\lvert\psi\rangle=\sum_a\psi_a\lvert a\rangle$ with
 $p_a=\lvert\psi_a\rvert^2$.  We call the TQFT contribution
 $\Gamma_{T^2}(\psi)$ the \emph{topological entanglement entropy} (TEE).
 In this description, the large-region entropy has the positive-magnitude
 convention
 \begin{equation}
  S_A(\psi)=\alpha(L_1+L_2)-\Gamma_{T^2}(\psi)+o(1).
  \label{eq:tee-sign}
 \end{equation}
Thus $\Gamma_{T^2}\geq0$ is the TEE magnitude subtracted from the area law,
whereas the signed constant term appearing in $S_A$ is
$-\Gamma_{T^2}$.  Some literature calls that signed constant the TEE; in that
terminology it is the negative of our $\Gamma_{T^2}$.  With our convention,
\begin{align}
 \Gamma_{T^2}(\psi)&=2\log\cD-\mathcal P(\psi),\nonumber\\
 \mathcal P(\psi)&=2\sum_a p_a\log d_a-\sum_a p_a\log p_a,
 \label{eq:torusentropy}
\end{align}
 where $\alpha$ is nonuniversal and $\cD$ is given by
 Eq.~\eqref{eq:total-dimension}
\cite{DongEtAl,ZhangEtAl,WenMatsuuraRyu,LoChang}.  Unless a base is shown
explicitly, all logarithms are natural.  The factor of two is geometric: a
simply connected planar region has one entanglement circle and TEE
$\log\cD$, whereas this cut has two circles, each contributing $\log\cD$ in
 the vacuum-flux state.  This factor does not enter the invariant comparison
 between $Z(T^3)$ and $Z(S^3)$; it appears only when that comparison is
 realized by this two-circle torus cut.

The identification with microscopic entropies requires care.  For the
Levin--Wen annular conditional-mutual-information definition, additional
``spurious'' contributions can survive at large region size and change under
local circuits of constant depth within the same phase
\cite{KimEtAlTEE,LevinTEE}.  Under the assumptions of those works, the anyon
value $\log\cD$ is a lower bound for that annular TEE.  Their inequality is
not established for cylinder extrapolation \cite{LevinTEE}, so we do not
apply it to the noncontractible torus cut here.  Equation~\eqref{eq:torusentropy}
specifies the TQFT contribution after nonuniversal boundary effects have been
separated.  Extracting it in a microscopic model requires control of any
additional constants.  Circuit minimization removes spurious contributions
for the reference-state class treated in Ref.~\cite{KimEtAlTEE}; no such
prescription is proved here for general WZW realizations.  Our optimization
varies the torus flux amplitudes within $\mathcal H(T^2)$, and does not itself
remove microscopic spurious contributions.

We next maximize Eq.~\eqref{eq:torusentropy} over $\psi$.  In a unitary fusion
category,
$d_a=d_{\bar a}>0$, and the vacuum occurs once in $a\times\bar a$.  Therefore
$d_a^2=\sum_c N_{a\bar a}^{\phantom{a\bar a}c}d_c\geq d_0=1$.  Since we use
a normalized state with $0 \le p_a \le 1$, therefore
$\mathcal P(\psi)=-\sum_a p_a\log p_a + 2\sum_a p_a\log d_a\geq0$.  For the opposite
bound,
introduce
the probability distribution $w_a=d_a^2/\cD^2$.  Positivity of relative
entropy\footnote{It's commonly called the Kullback-Leibler divergence.} yields
\begin{equation}
 0\leq D_{\rm KL}(p\Vert w)
 =\sum_a p_a\log\frac{p_a}{w_a}
 =2\log\cD-\mathcal P(\psi),
 \label{eq:penalty-upper}
\end{equation}
and hence
\begin{equation}
 0\leq\mathcal P(\psi)\leq2\log\cD,
 \qquad 0\leq\Gamma_{T^2}(\psi)\leq2\log\cD.
 \label{eq:penalty-bounds}
\end{equation}
Consequently the first-stage optimization is
\begin{equation}
 \max_{\substack{\psi\in\mathcal H(T^2)}}
 \Gamma_{T^2}(\psi)=2\log\cD.
 \label{eq:state-maximization}
\end{equation}
Indeed, $\mathcal P(\psi)=0$ exactly when $p$ is concentrated on one simple
object $a$ with $d_a=1$.  The vacuum object has $d_0=1$, so the vacuum-flux
state $\lvert0\rangle$ is always a maximizer.  More generally, all and only
definite Abelian-flux states maximize $\Gamma_{T^2}$.  For a definite flux,
\(
 \Gamma_{T^2}^{(a)}=2\log(\cD/d_a),
\)
so every non-Abelian flux ($d_a>1$) gives a strictly smaller TEE.  We
use the vacuum as the canonical maximizing representative below.

At the opposite endpoint, $\mathcal P=2\log\cD$ and $\Gamma_{T^2}=0$ exactly
when $p_a=w_a=d_a^2/\cD^2$; the relative phases are arbitrary.  One
representative is the Kirby color state
\begin{equation*}
 \lvert \Omega\rangle=\cD^{-1}\sum_a d_a\lvert a\rangle.
\end{equation*}

In an Abelian category every quantum dimension is one and
$\cD=\sqrt{r(\cC)}$.  Equation~\eqref{eq:torusentropy} becomes
\begin{equation}
 \Gamma_{T^2}^{\rm Abelian}(\psi)
 =\log r(\cC)+\sum_a p_a\log p_a
 \leq\log r(\cC)=(\log 2)q(\cC).
 \label{eq:abelian}
\end{equation}
The maximal Abelian value therefore grows linearly with binary logarithmic
rank.
Growth on the quadratic scale must come from non-Abelian quantum dimensions.
After the state optimization in Eq.~\eqref{eq:state-maximization}, the WZW
problem is precisely to maximize $2\log\cD$ subject to
$r(\mathfrak g,k)\leq R$.

\section{WZW Families and Main Result}
\label{sec:landscape}

A WZW category $\cC(\mathfrak g,k)$ is specified by a simple Lie algebra
$\mathfrak g$ and a positive integer level $k$.  Its simple objects are the
integrable highest-weight representations of the affine Lie algebra $\widehat{\mathfrak g}_k$.
These labels describe Wilson lines and torus ground states in the simply
connected Chern--Simons theory and primary fields in the chiral boundary WZW
model; the same modular $S$ matrix acts in each description
\cite{Witten,KacPeterson}.  The four classical families are
\begin{equation}
 \begin{aligned}
  A_{n-1}=\mathfrak{su}_n&\ \longleftrightarrow\ SU(n)_k,\\
  B_n=\mathfrak{so}_{2n+1}&\ \longleftrightarrow\ \Spin(2n+1)_k,\\
  C_n=\mathfrak{sp}_{2n}&\ \longleftrightarrow\ Sp(2n)_k,\\
  D_n=\mathfrak{so}_{2n}&\ \longleftrightarrow\ \Spin(2n)_k.
 \end{aligned}
 \label{eq:dictionary}
\end{equation}
In this dictionary, $SU(n)$ is the group of $n\times n$ unitary matrices with
determinant one, and $\mathfrak{su}_n$ is its Lie algebra; this is the
$A_{n-1}$ family.  The $B_n$ and $D_n$ series are respectively
the odd- and even-dimensional orthogonal families.  Here $SO(N)$ is the
rotation group of $\mathbb R^N$, $\mathfrak{so}_N$ is its algebra of
infinitesimal rotations, and $\Spin(N)$ is its simply connected double cover.
Finally, $C_n$ is the compact symplectic family: its group consists of the
unitary transformations of $\mathbb C^{2n}$ that preserve the standard
symplectic form.  We denote this group by $Sp(2n)$,
and its Lie algebra by $\mathfrak{sp}_{2n}$.  In terms of the generic
Lie-rank symbol introduced above, $\ell=n$ for $B_n,C_n,D_n$, whereas
$\ell=n-1$ for $A_{n-1}$.\footnote{The rank of a simple Lie algebra is the dimension of its Cartan subalgebra, or equivalently the number of simple roots in its root system.}

The remaining simple Lie
algebras, $G_2,F_4,E_6,E_7,E_8$, are called exceptional; unlike the classical
families, their ranks are fixed and do not grow with a parameter. The ranks of the exceptional algebras are $2,4,6,7,8$, respectively.

The finite-dimensional complex simple Lie algebras are exhausted by $A_{n-1}$ for $n\geq2$, $B_n$ for $n\geq3$, $C_n$ for $n\geq2$, $D_n$ for $n\geq4$ and the exceptional algebras $G_2,F_4,E_6,E_7,E_8$.  There are some low-rank coincidences: $B_1=C_1=A_1$, $B_2=C_2$, and $D_3=A_3$.
The remaining classical symbol $D_2 = A_1\oplus A_1$ is semisimple rather
than simple, so it is not part of the simple envelope.
We will deal with semisimple algebras in the separate stacked extension, see Corollary~\ref{cor:stacking}.
%   With independent
% levels $k_1,k_2$, it is already included in the class of
% Eq.~\eqref{eq:semisimple-factorization}:
% \begin{equation}
%  \widehat{\mathfrak{so}(4)}_{(k_1,k_2)}
%  \cong\widehat{\mathfrak{su}(2)}_{k_1}\oplus
%        \widehat{\mathfrak{su}(2)}_{k_2},
%  \qquad
%  \cC(D_2;k_1,k_2)
%  \simeq\cC(A_1,k_1)\boxtimes\cC(A_1,k_2).
%  \label{eq:D2-factorization}
% \end{equation}

The categorical rank $r(\mathfrak g,k)$ counts lattice points in the level-$k$
affine Weyl alcove.  The total quantum dimension is
$\cD(\mathfrak g,k)=1/S_{00}$, and the Kac--Peterson formula expresses $S_{00}$ as a product
over the positive roots of $\mathfrak g$.
Appendix~\ref{app:wzw-data} collects the modular-category data and obtains the
normalized finite product from the Weyl denominator, while
Appendix~\ref{app:wzw-counts} derives the sector generating functions from the
comark constraint.

We call a sequence balanced when the parameters exchanged by rank-level
duality are asymptotically equal: $k/n\to1$ for types $A,C$ and $k/N\to1$
for the orthogonal theories $\Spin(N)_k$.  In the counting model below this is
the half-filled point.

\begin{theorem}[Simple WZW cutoff envelope]
\label{thm:main}
Let $\mathfrak g$ range over the nonduplicated finite-dimensional complex
simple Lie algebras just specified, let $k\in\mathbb Z_{>0}$, and use the
simply connected untwisted WZW modular category $\cC(\mathfrak g,k)$.
Let $r(\mathfrak g,k)$ be its rank and $\cD(\mathfrak g,k)$ its total quantum dimension.
Then the envelope $F_{\WZW}$ defined by
\begin{equation}
 F_{\WZW}(R) = \sup_{\substack{\mathfrak g\ \mathrm{simple},\ k\in\mathbb Z_{>0}\\
                  r(\mathfrak g,k)\leq R}}
 2\log\cD(\mathfrak g,k),
\end{equation}
satisfies the asymptotic law
\begin{equation}
 \lim_{R\to\infty}\frac{F_{\WZW}(R)}{(\log_2R)^2}
 =\frac{7\zeta(3)}{4\pi^2}.
 \label{eq:mainlaw}
\end{equation}
Moreover, with $r_n=r(C_n,n)=\binom{2n}{n}$,
\begin{equation}
 \lim_{n\to\infty}
 \frac{2\log\cD(C_n,n)}{(\log_2r_n)^2}
 =\frac{7\zeta(3)}{4\pi^2},
 \label{eq:C-saturates}
\end{equation}
so $Sp(2n)_n$ is an asymptotically saturating sequence.  If
$N_j,k_j\to\infty$ with $k_j/N_j\to1$, the allowed
$\Spin(N_j)_{k_j}$ sequences have the same leading ratio.
%   This leading-order
% tie does not assert that an orthogonal theory, or any particular theory, is a
% unique maximizer at finite $R$.
\end{theorem}

% The coefficient $7\zeta(3)/(4\pi^2)=0.2131391994\ldots$ arises from the odd
% Fourier modes of the half-filled root interaction.
%  Natural logarithms are used in $F_{\WZW}$ and binary logarithms only set the horizontal units.

The type-$A$ family $SU(n)_k$ reaches half this coefficient.  The source of
the relative factor is summarized in Table~\ref{tab:proportional-summary}.
Type $A$ has only the pair roots $e_i-e_j$, whereas $B,C,D$ also have the
$e_i+e_j$ channel.  Each channel contains $O(n^2)$ roots.  The $n$ short roots
of $B_n$ and the $n$ long roots of $C_n$ contribute only $O(n\log n)$.
For $X\in\{A,B,C,D\}$, write $r_X(n,k)$ and $\cD_X(n,k)$ for the sector
count and total quantum dimension of the corresponding family in
Eq.~\eqref{eq:dictionary}.  Arguments $(n,k)$ are suppressed when no
ambiguity can arise; $K,t$, and $\mu_X$ are defined in the table caption.
\begin{table}[t]
\centering
\caption{Data controlling proportional classical limits.  Here
$K=k+h_X^\vee$, with $h_X^\vee$ the dual Coxeter number; $t$ is the effective
rank-level ratio shown in the table; and $\mu_X$ counts the leading pair-root
channels.
The functions $H$ and $J$ are defined in Eq.~\eqref{eq:HJ}; $n$ is
the displayed family parameter, not always the exact Lie rank.
}
\label{tab:proportional-summary}
\begin{tabular}{c c c c c}
\toprule
family & shifted level $K$ & effective $t$ & $O(n^2)$ positive roots & $\mu_X$\\
\midrule
$A_{n-1}$ & $k+n$ & $k/n$ & $e_i-e_j$ & $1$\\
$B_n$ & $k+2n-1$ & $k/(2n)$ & $e_i-e_j,\ e_i+e_j$ & $2$\\
$C_n$ & $k+n+1$ & $k/n$ & $e_i-e_j,\ e_i+e_j$ & $2$\\
$D_n$ & $k+2n-2$ & $k/(2n)$ & $e_i-e_j,\ e_i+e_j$ & $2$\\
\bottomrule
\end{tabular}
\end{table}
% Uniformly for $t$ in compact subsets of $(0,\infty)$,
These data give
$\log r_X=nH(t)+O(\log n)$ and
$\log\cD_X=\mu_Xn^2J(t)+O(n\log n)$; the density derivation and all
normalization factors are given in Appendix~\ref{app:wzw-asymptotics}.

If $c_X$ is the largest limiting value of $\Gamma_{T^2}/q(\cC)^2$ along proportional trajectories of type $X$, then
\begin{equation}
 c_A=\frac{7\zeta(3)}{8\pi^2},
 \qquad c_B=c_C=c_D=\frac{7\zeta(3)}{4\pi^2}.
 \label{eq:familycoefficients}
\end{equation}
Half filling uniquely maximizes the proportional-limit coefficient, as proved
in Section~\ref{sec:global}.  Fixed-rank and rank-level estimates control the
exceptional and unbalanced sequences.

\begin{figure}[t]
 \centering
 \includegraphics[width=0.8\textwidth]{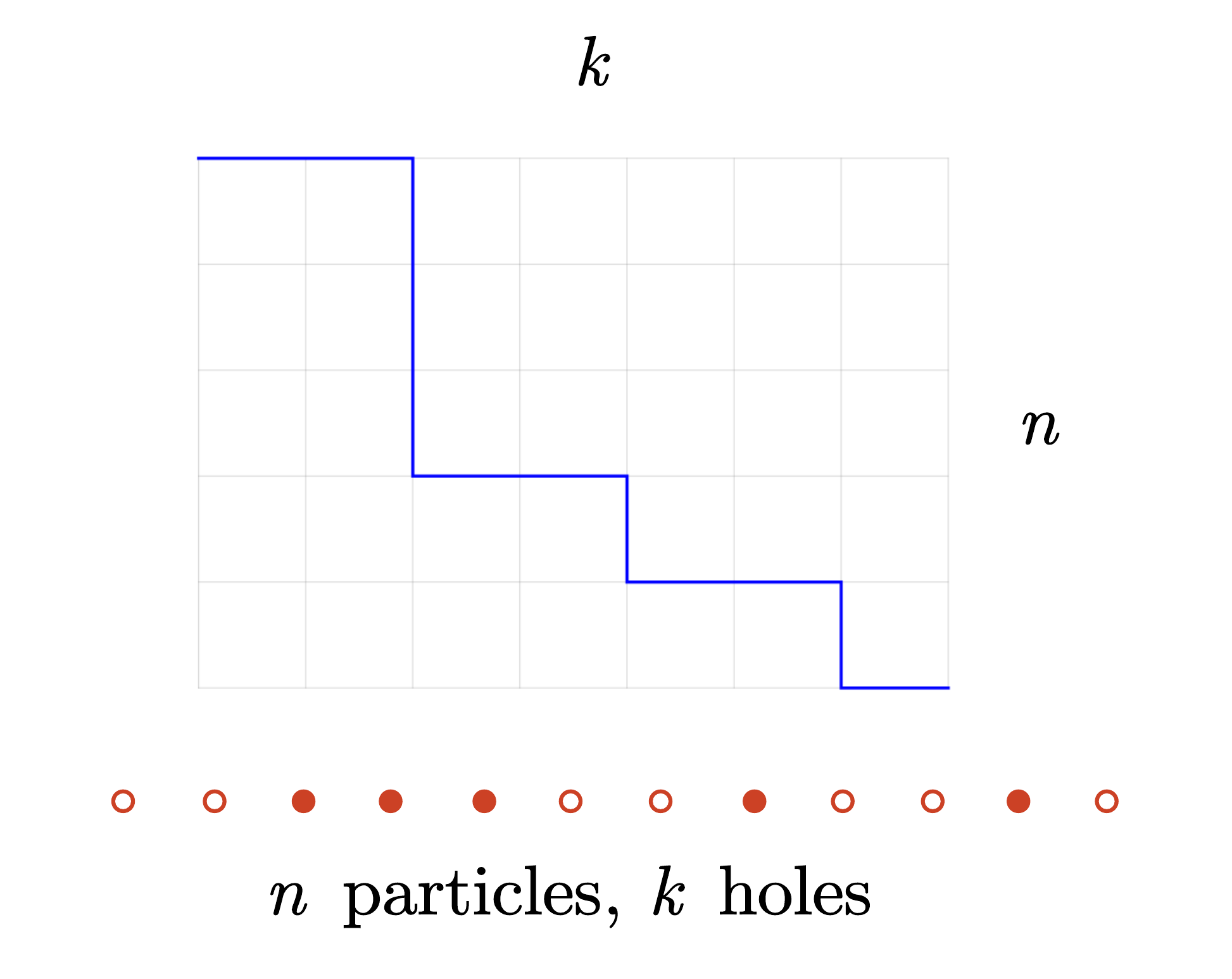}
 \caption{Integrable $Sp(2n)_k$ sectors correspond to lattice paths in an
 $n\times k$ rectangle, or to $n$ occupied sites among $n+k$ ordered sites.
 Transposition exchanges occupied and empty sites, so $n=k$ is half filling.
 The path shown has $(n,k)=(5,7)$.}
 \label{fig:origin}
\end{figure}

\section{Half Filling and the Odd-Mode Constant}
\label{sec:half}

We first analyze type $C$ at arbitrary rank and level.

\emph{Sector-counting model.}
An integrable highest weight of $C_n$ at level $k$ is a Young diagram in an
$n\times k$ rectangle.  Reading its boundary as a binary word, with vertical
steps occupied and horizontal steps empty, gives a bijection with $n$-element
subsets of $n+k$ ordered sites. We show this counting model in Figure~\ref{fig:origin}. Hence
\begin{equation}
 r_{n,k}=r(C_n,k)=\binom{n+k}{n}.
 \label{eq:binomial}
\end{equation}
The occupation language borrows terminology from exclusion statistics
\cite{HaldaneExclusion,WuExclusion}, but these particles and holes are only
steps of the boundary path.  The type-$C$ comark calculation in
Appendix~\ref{app:wzw-counts} gives the same bijection.  Its filling fraction is
\begin{equation}
 s=\frac{n}{n+k}.
 \label{eq:filling}
\end{equation}
Exchanging $n$ and $k$ interchanges vertical and horizontal steps, so type-$C$
rank-level exchange is particle--hole exchange, $s\leftrightarrow1-s$.
At this stage no balance condition has been imposed.

\emph{Kac--Peterson formula.}
% The total quantum dimension is $\cD=1/S_{00}$.
For type $C_n$, use
$(e_i,e_j)=\delta_{ij}/2$, so the long roots $2e_i$ have squared length two.
The Weyl vector is $\rho=\sum_{i=1}^n(n+1-i)e_i$.  Its angular coordinates are
\begin{equation}
 \theta_i=\frac{\pi(\rho,e_i)}{k+n+1}
 =\frac{\pi(n+1-i)}{2(k+n+1)},
 \qquad i=1,\ldots,n.
 \label{eq:Cangles}
\end{equation}
The finite Weyl denominator is
\begin{align}
 S_{00}&=2^{-n/2}(k+n+1)^{-n/2}\Delta_{C_n}(\bm\theta),\nonumber\\
 \Delta_{C_n}(\bm\theta)
 &=\prod_{i<j}2\sin(\theta_i-\theta_j)\,2\sin(\theta_i+\theta_j)
   \prod_i2\sin(2\theta_i).
 \label{eq:weylC}
\end{align}
This is the Kac--Peterson formula \eqref{S:S00} specialized to the positive
roots $e_i-e_j$, $e_i+e_j$, and $2e_i$ of $C_n$.  In the eigenvalue-gas
interpretation the $\theta_i$ are angular positions; algebraically they are
coordinates of the Weyl vector and are unrelated to the auxiliary occupation
sites above.

Equation~\eqref{eq:weylC} now makes the quadratic scaling visible.  Every
unordered pair $i<j$ occurs in two interaction channels, one through
$\theta_i-\theta_j$ and one through $\theta_i+\theta_j$.  Thus the Weyl
determinant contains $2\binom n2=O(n^2)$ pair factors; the remaining $n$
long-root factors are lower order.  After taking $-\log S_{00}$, the pair
terms have the form of an eigenvalue-gas interaction
\cite{Witten,KacPeterson,Marino}.

To compare rank and level on the same scale, set
\begin{equation}
 t=\frac{k}{n},\qquad s=\frac1{1+t}=\frac{n}{n+k}.
\end{equation}
For fixed $t>0$, the uniform root-height estimate in
Appendix~\ref{app:wzw-asymptotics}, together with Stirling's formula, gives
\begin{align}
 \log r_{n,k}&=nH(t)+O(\log n),\nonumber\\
 \log\cD(C_n,k)&=2n^2J(t)+O(n\log n),
 \label{eq:continuum}
\end{align}
where
\begin{align}
 H(t)&=(1+t)\log(1+t)-t\log t,\nonumber\\
 J(t)&=-\int_0^1(1-x)
 \log\!\left(2\sin\frac{\pi x}{1+t}\right)\dd x.
 \label{eq:HJ}
\end{align}
% The first line is the configuration entropy of the boundary words.
The difference- and sum-root height densities combine into the weight $1-x$ in
the second line, and the two pair channels produce the prefactor two.  Hence
$\log r=O(n)$ whereas
$\log\cD=O(n^2)$, which already explains why the torus TEE can be
quadratic in the logarithmic rank $q(\cC)=(\log r)/\log 2$.

The Fourier series
$-\log(2\sin \pi sx)=\sum_{m\geq1}\cos(2\pi msx)/m$ and the elementary
integral
\begin{equation}
 \int_0^1(1-x)\cos(2\pi msx)\,\dd x
 =\frac{1-\cos(2\pi ms)}{(2\pi ms)^2}
 \label{eq:pairintegral}
\end{equation}
give the spectral form
\begin{align}
 H(t)&=\frac{h(s)}s,&
 h(s)&=-s\log s-(1-s)\log(1-s),\nonumber\\
 J(t)&=\frac{A(s)}{4\pi^2s^2},&
 A(s)&=\sum_{m=1}^{\infty}\frac{1-\cos(2\pi ms)}{m^3}
 =\zeta(3)-\Re\Li_3(e^{2\pi\ii s}).
 \label{eq:spectral}
\end{align}
Here $\Li_3(z)=\sum_{m\geq1}z^m/m^3$ is the trilogarithm
\cite{Lewin} and $\Re$ denotes the real part.  The same function occurs in the nonpolynomial genus-zero part
of the large-$N$ Chern--Simons free energy on $S^3$.  That term is
conventionally written as $\Li_3(e^{-t})$; after
$t\mapsto-2\pi\ii s$, the function used here is its sign-reversed,
$\zeta(3)$-shifted real part
\cite{GopakumarVafa,SinhaVafa,MarinoReview}.  Equation~\eqref{eq:spectral},
however, follows directly from the Kac--Peterson root product.

\emph{Balanced specialization.}
Set $k=n$.  Then $s=1/2$, and the counting model has $n$ occupied sites among
$2n$ sites.  Thus
\begin{equation}
 1-\cos(2\pi ms)=1-(-1)^m
 =\begin{cases}0,&m\text{ even},\\2,&m\text{ odd}.
 \end{cases}
 \label{eq:oddprojector}
\end{equation}
The even modes vanish and the odd modes are doubled:
\begin{equation}
 A(1/2)=2\sum_{m\,\mathrm{odd}}\frac1{m^3}
 =2(1-2^{-3})\zeta(3)=\frac{7}{4}\zeta(3).
 \label{eq:oddconstant}
\end{equation}
Write $r_n=\binom{2n}{n}$ and $q_n=\log_2r_n$.  Combining
Eqs.~\eqref{eq:continuum}, \eqref{eq:spectral}, and
\eqref{eq:oddconstant} with $\Gamma_{T^2}=2\log\cD$ gives
\begin{equation}
 \frac{\Gamma_{T^2}}{q_n^2}\longrightarrow
 \frac{A(1/2)}{\pi^2}=\frac{7\zeta(3)}{4\pi^2}.
 \label{eq:ratioodd}
\end{equation}
The coefficient is the odd-integer part of the cubic zeta sum selected by
half filling.  The power $m^{-3}$ arises because the Fourier coefficient of the
logarithmic Weyl kernel contributes $m^{-1}$, while the two integrations in
Eq.~\eqref{eq:pairintegral} contribute $m^{-2}$.

\section{Global Optimality}
\label{sec:global}

Appendices~\ref{app:wzw-counts} and~\ref{app:wzw-asymptotics} combine the
comark counts and Stirling's formula
with a root-height Riemann sum, uniformly for $t$ in compact subsets of
$(0,\infty)$, to give
\begin{align}
 \log r_X&=nH(t)+O(\log n),\nonumber\\
 \log\cD_X&=\mu_X n^2J(t)+O(n\log n).
 \label{eq:familyreduction}
\end{align}
Here $X$ can be $A$, $B$, $C$, or $D$.  For $A$ and $C$ the ratio is $t=k/n$; for $B$ and $D$ the natural effective
ratio is $t=k/(2n)$.  Here $\mu_A=1$ and
$\mu_B=\mu_C=\mu_D=2$.  The common function $H$ comes from the binomial entropy
of integrable weights.  The common function $J$ comes from the continuum root
interaction.  Type $A$ has only the $e_i-e_j$ pair channel, while the other
three families have both $e_i-e_j$ and $e_i+e_j$, which is the origin of
$\mu_X$.

Put $q_X=\log_2r_X$.  Using the spectral representation in
Eq.~\eqref{eq:spectral}, Eq.~\eqref{eq:familyreduction} becomes
\begin{equation}
 \lim_{n\to\infty}\frac{\Gamma_{T^2}}{q_X^2}
 =
 \begin{cases}
  \frac12\mathcal E(s),&X=A,\\[2pt]
  \mathcal E(s),&X=B,C,D,
 \end{cases}
 \label{eq:efficiency}
\end{equation}
where $s=1/(1+t)$ and
\begin{equation}
 \mathcal E(s)=\frac{(\log 2)^2}{\pi^2}\frac{A(s)}{h(s)^2}.
 \label{eq:efficiencyfunction}
\end{equation}
The proportional problem is therefore to maximize $A(s)/h(s)^2$; the only
family dependence is the factor of two in Eq.~\eqref{eq:efficiency}.

\subsection{Proportional regime: Lie rank and level grow at fixed ratio}
\label{sec:sharp-proof}

\begin{lemma}[Sharp entropy--spectral inequality]
\label{lem:sharp}
For $0<s<1$,
\begin{equation}
 A(s)\leq \frac{7\zeta(3)}{4(\log 2)^2}h(s)^2,
 \label{eq:sharp}
\end{equation}
with equality if and only if $s=1/2$.
\end{lemma}

\begin{proof}[Proof of Lemma~\ref{lem:sharp}]
Put $y=1-2s$ and $z=y^2$.  Particle--hole symmetry makes the
entropy and spectral energy even in $y$, so they may be written as functions
of $z\in[0,1)$:
\begin{align}
 \mathsf h(z)
 &=\log 2-\frac12[(1+y)\log(1+y)+(1-y)\log(1-y)]\nonumber\\
 &=\log 2-\sum_{m\geq1}\frac{z^m}{2m(2m-1)},\label{S:hseries}\\
 \mathsf A(z)&=\zeta(3)-\Re\Li_3(-e^{-\ii\pi y}).
\end{align}
Twice differentiating the last expression gives
\begin{equation}
 \frac{\dd^2\mathsf A}{\dd y^2}
 =-\pi^2\log\left(2\cos\frac{\pi y}{2}\right).
\end{equation}
For $|y|<1$, integrate the Taylor series twice.  The two integration
constants are fixed by $\mathsf A(0)=7\zeta(3)/4$ and
$\partial_y\mathsf A(0)=0$, the latter following from evenness.  This gives
\begin{equation}
 \mathsf A(z)=\frac74\zeta(3)-\frac{\pi^2\log 2}{2}z
 +\sum_{m\geq1}
 \frac{\pi^2(1-4^{-m})\zeta(2m)}
 {m(2m+1)(2m+2)}z^{m+1}.
 \label{S:Aseries}
\end{equation}

Let $C_\star=7\zeta(3)/[4(\log 2)^2]$ and define
\begin{equation}
 \frac{\mathsf h(z)}{1-z}=\sum_{n\geq0}u_n z^n,
 \qquad
 \frac{\mathsf A(z)}{(1-z)^2}=\sum_{n\geq0}v_n z^n.
\end{equation}
Equation~\eqref{S:hseries} implies the positive integral representation
\begin{equation}
 u_n=\int_0^1\frac{x^{2n}}{1+x}\dd x>0.
 \label{S:un}
\end{equation}
Define the coefficients \(c_n\) by
\begin{equation}
 \frac{C_\star\mathsf h(z)^2-\mathsf A(z)}{(1-z)^2}
 =\sum_{n\geq0}c_n z^n,
 \qquad
 c_n=C_\star\sum_{i=0}^n u_i u_{n-i}-v_n,
 \label{S:cn}
\end{equation}
The constant term vanishes because
$C_\star\mathsf h(0)^2=\mathsf A(0)$.

We first write the remaining coefficients in exact form.  Define the rational
numbers and polynomial
\begin{align}
 \sigma_j&=\sum_{m=1}^j\frac1{2m(2m-1)},\qquad \sigma_0=0,\nonumber\\
 S_n&=\sum_{j=0}^n\sigma_j,\qquad
 T_n=\sum_{j=0}^n\sigma_j\sigma_{n-j},\nonumber\\
 \beta_m&=\frac{(1-4^{-m})|B_{2m}|2^{2m-1}}
 {m(2m+1)(2m+2)(2m)!},\nonumber\\
 W_n(p)&=\sum_{m=1}^{n-1}(n-m)\beta_m p^{m+1}.
 \label{S:algebraic-data}
\end{align}
The sum defining \(W_1\) is empty.  Here \(B_r\) is the \(r\)-th
\emph{Bernoulli number}, defined by the generating function
\begin{equation}
 \frac{t}{e^t-1}=\sum_{r\geq0}B_r\frac{t^r}{r!}.
 \label{S:Bernoulli-definition}
\end{equation}
For example, \(B_2=1/6\) and \(B_4=-1/30\).  More generally,
\((-1)^{m+1}B_{2m}>0\) for \(m\geq1\), and Euler's formula reads
\begin{equation}
 \zeta(2m)
 =(-1)^{m+1}\frac{B_{2m}(2\pi)^{2m}}{2(2m)!}
 =\frac{|B_{2m}|(2\pi)^{2m}}{2(2m)!}.
 \label{S:Euler-even-zeta}
\end{equation}
The definition and Euler's identity are standard; see
Ref.~\cite{NISTHandbook}.
This explains both the Bernoulli number and the absolute value in
\(\beta_m\).

We now derive \(c_n\) step by step.  First, multiplying the entropy series
in Eq.~\eqref{S:hseries} by
\((1-z)^{-1}=\sum_{r\geq0}z^r\) gives
\begin{equation}
 u_j=[z^j]\frac{\mathsf h(z)}{1-z}
 =\log 2-\sigma_j.
 \label{S:un-exact}
\end{equation}
Consequently,
\begin{align}
 \sum_{i=0}^n u_i u_{n-i}
 &=\sum_{i=0}^n
   (\log 2-\sigma_i)(\log 2-\sigma_{n-i})\nonumber\\
 &=(n+1)(\log 2)^2-2(\log 2)S_n+T_n.
 \label{S:u-convolution-exact}
\end{align}

Second, let \(a_j=[z^j]\mathsf A(z)\).  Substitution of
Eq.~\eqref{S:Euler-even-zeta} into Eq.~\eqref{S:Aseries} gives
\begin{equation}
 a_0=\frac74\zeta(3),\qquad
 a_1=-\frac{\pi^2\log 2}{2},\qquad
 a_{m+1}=\beta_m(\pi^2)^{m+1}\quad(m\geq1).
 \label{S:A-coefficients-exact}
\end{equation}
Since
\((1-z)^{-2}=\sum_{r\geq0}(r+1)z^r\), the Cauchy product gives
\begin{align}
 v_n
 &=[z^n]\frac{\mathsf A(z)}{(1-z)^2}
  =\sum_{j=0}^n(n-j+1)a_j\nonumber\\
 &=\frac{7(n+1)\zeta(3)}4
   -\frac{n\pi^2\log 2}{2}
   +W_n(\pi^2).
 \label{S:vn-exact}
\end{align}

Finally, insert Eqs.~\eqref{S:u-convolution-exact} and
\eqref{S:vn-exact} into Eq.~\eqref{S:cn}.  The terms
\(7(n+1)\zeta(3)/4\) cancel because
\(C_\star=7\zeta(3)/[4(\log 2)^2]\), leaving the exact expression
\begin{equation}
 c_n=\frac{n\pi^2\log 2}{2}-W_n(\pi^2)
 -\frac{7S_n\zeta(3)}{2\log 2}
 +\frac{7T_n\zeta(3)}{4(\log 2)^2}.
 \label{S:cn-exact}
\end{equation}
For \(1\leq n\leq10\), direct exact substitution into
Eq.~\eqref{S:cn-exact} gives
\begin{equation}
 c_n>0.
 \label{S:cn-finite}
\end{equation}
This is a finite symbolic check of the displayed expression: \(S_n,T_n\)
and the coefficients of \(W_n\) are rational, and no decimal approximation
is used.

It remains to prove positivity for \(n\geq11\).  From Eq.~\eqref{S:un},
\(u_i\geq[2(2i+1)]^{-1}\).  Applying this bound to both factors and using
partial fractions gives
\begin{equation}
 \sum_{i=0}^n u_i u_{n-i}
 \geq\frac1{4(n+1)}\sum_{i=0}^n\frac1{2i+1}
 \geq\frac{\log(2n+3)}{8(n+1)}.
 \label{S:conv}
\end{equation}

To control $v_n$, the positive spectral series gives
\begin{equation}
 A(s)=O\left(s^2\log\frac1s\right)\qquad(s\downarrow0).
\end{equation}
Use $1-\cos x\leq\min\{2,x^2/2\}$ and split the sum at
$m=\lfloor1/s\rfloor$.  Since $1-z=4s(1-s)$, this gives
$\mathsf A(z)=O((1-z)^2\log[1/(1-z)])$.  The coefficients in
Eq.~\eqref{S:Aseries} are $O(m^{-3})$, so the series and its first derivative
converge absolutely at $z=1$.  Thus
\begin{equation}
 \sum_{j\geq0}a_j=0,
 \qquad
 \sum_{j\geq0}j a_j=0,
\end{equation}
where $a_j$ are the coefficients of $\mathsf A$; these identities are
$\mathsf A(1)=\mathsf A'(1)=0$.  Consequently
\begin{equation}
 v_n=\sum_{j>n}(j-n-1)a_j.
\end{equation}
Using $(1-4^{-m})\zeta(2m)\leq\pi^2/8$ and
$(2m+1)(2m+2)\geq4m^2$ yields
\begin{align}
 v_n&\leq\frac{\pi^4}{32}
 \sum_{m=n+1}^{\infty}\frac{m-n}{m^3}\nonumber\\
 &\leq\frac{\pi^4}{32}
 \left(\frac1{2n}+\frac4{27n^2}\right).
 \label{S:vn}
\end{align}
For the second inequality, compare the unimodal function $(x-n)/x^3$ with
its integral over $[n,\infty)$ and add its maximum $4/(27n^2)$.
Combining Eqs.~\eqref{S:conv} and~\eqref{S:vn} yields the explicit estimate
\begin{equation}
 c_n\geq
 \frac{C_\star\log(2n+3)}{8(n+1)}
 -\frac{\pi^4}{32}\left(\frac1{2n}+\frac4{27n^2}\right).
 \label{S:cn-tail}
\end{equation}

Define a function
\begin{equation}
 f(x):=\frac{C_\star\log(2x+3)}{8(x+1)}
 -\frac{\pi^4}{32}\left(\frac1{2x}+\frac4{27x^2}\right)
 \label{S:cn-function}
\end{equation}
for $x\geq1$.  The right-hand side of Eq.~\eqref{S:cn-tail} is $f(n)$.
Differentiating \(xf(x)\) gives
\begin{equation}
 \frac{\dd}{\dd x}\bigl[xf(x)\bigr]
 =\frac{C_\star}{8}
 \left[
  \frac{\log(2x+3)}{(x+1)^2}
  +\frac{2x}{(x+1)(2x+3)}
 \right]
 +\frac{\pi^4}{216x^2}>0.
 \label{S:cn-monotone}
\end{equation}
Thus \(xf(x)\) is increasing for \(x\geq1\).  The rational bounds
\(\zeta(3)>6/5\), \(\log2<7/10\), \(\log25>16/5\), and \(\pi<22/7\)
give the exact estimate \(f(11)>139/259308>0\).
For \(x\geq11\), it follows that
\(xf(x)\geq11f(11)>0\), hence \(f(x)>0\).  Since
\(c_n\geq f(n)\), we conclude that \(c_n>0\) for every \(n\geq11\).
Together with Eq.~\eqref{S:cn-finite}, this proves \(c_n>0\) for all
\(n\geq1\).
Therefore
\begin{equation}
 C_\star\mathsf h(z)^2-\mathsf A(z)
 =(1-z)^2\sum_{n\geq1}c_n z^n>0
\end{equation}
for $0<z<1$.  Equality occurs at $z=0$, or $s=1/2$, which proves
Eq.~\eqref{eq:sharp}.
\end{proof}

Combining Eqs.~\eqref{eq:efficiency} and~\eqref{eq:efficiencyfunction} with
Lemma~\ref{lem:sharp} proves the proportional coefficients stated in
Eq.~\eqref{eq:familycoefficients}.  The strict equality condition shows that
the maximizing ratio is uniquely \(t=1\) in each classical family; types
\(B,C,D\) attain the coefficient in Theorem~\ref{thm:main}, while type \(A\)
attains half of it.  It remains to control trajectories approaching the
boundary of the proportional regime.

To make the exhaustion explicit, let $\ell$ denote the Lie rank.  After
passing to a subsequence, every classical sequence with categorical rank
tending to infinity falls into exactly one of four regimes: both $\ell$ and
$k$ diverge with a finite nonzero limiting ratio; both diverge while
$k/\ell$ tends to zero or infinity; $\ell$ is fixed and $k\to\infty$; or
$k$ is fixed and $\ell\to\infty$.  The first regime has just been treated.
The next three subsections handle the remaining cases.  An exceptional
sequence varies only in its level, since its Lie algebra is one of the five
fixed exceptional algebras, and therefore belongs to the fixed-Lie-algebra
regime.

\subsection{Fixed Lie rank and growing level; exceptional algebras}
\label{sec:fixed-rank-regime}

\begin{lemma}[Fixed Lie algebra and exceptional families]
\label{lem:fixed-rank}
For every fixed simple Lie algebra $\mathfrak g$ (so its Lie rank does not
vary with $k$),
\begin{equation}
 \frac{\log\cD(\mathfrak g,k)}
 {[\log r(\mathfrak g,k)]^2}\longrightarrow0
 \qquad (k\to\infty).
 \label{S:fixedrankzero}
\end{equation}
In particular, every exceptional family has zero quadratic coefficient.
\end{lemma}

\begin{proof}
Equations~\eqref{S:fixedrankcount} and \eqref{S:fixedrankdimension} give
$\log r(\mathfrak g,k)=\operatorname{rank}(\mathfrak g)\log k+O(1)$ and
$\log\cD(\mathfrak g,k)=\tfrac12\dim(\mathfrak g)\log k+O(1)$.
Their quotient has the limit in Eq.~\eqref{S:fixedrankzero}; the five
exceptional cases are recorded explicitly in
Table~\ref{tab:exceptional-asymptotics}.
\end{proof}

\subsection{Fixed level and growing Lie rank}
\label{sec:fixed-level-regime}

\begin{lemma}[Low orthogonal levels]
\label{lem:low-orthogonal}
For every integer $N\geq5$,
\begin{align}
 \cD(\mathfrak{so}_N,1)&=2,&
 \cD(\mathfrak{so}_N,2)&=2\sqrt N,\nonumber\\
 \cD(\mathfrak{so}_N,3)&=
 \frac{\sqrt{N+1}}{\sin[\pi/(2(N+1))]},&
 \cD(\mathfrak{so}_N,4)&=
 \frac{N+2}{2\sin^2[\pi/(N+2)]}.
 \label{S:lowO}
\end{align}
Consequently, along every fixed-level sequence with $1\leq k\leq4$,
either the sector count is bounded or
$\log\cD(\mathfrak{so}_N,k)=O(\log r(\mathfrak{so}_N,k))$.
\end{lemma}

\begin{proof}
The raw Kac--Peterson products are evaluated in
Table~\ref{tab:low-orthogonal}; the exact sector counts in
Eq.~\eqref{S:low-orthogonal-counts} give the stated comparison with
$\log r$.
\end{proof}

\begin{lemma}[Fixed-level classical limits]
\label{lem:fixed-level}
Fix a positive affine level $k$ and one of the classical families
$X=A,B,C,D$.  Along any sequence for which the Lie rank $\ell\to\infty$ and
$r_X\to\infty$,
\begin{equation}
 \log r_X=d_X(k)\log\ell+O(1),
 \qquad
 \log\cD_X=O(\log\ell),
 \label{S:fixedlevel-growth}
\end{equation}
where
$d_A(k)=d_C(k)=k$ and
$d_B(k)=d_D(k)=\lfloor k/2\rfloor$ for $k\geq2$.
% At orthogonal level one the sector counts are bounded, so those sequences are
% excluded by the hypothesis $r_X\to\infty$.
Consequently,
\begin{equation}
 \frac{\log\cD_X}{(\log r_X)^2}\longrightarrow0.
 \label{S:fixedlevelzero}
\end{equation}
\end{lemma}

\begin{proof}
We first verify the sector-count statement.  At fixed $k$, the two binomial
formulas in Eq.~\eqref{S:rAC} are polynomials of degree $k$ in $n$.
In the exact sums Eqs.~\eqref{S:rank-B-sum} and~\eqref{S:rank-D-sum}, the
largest summation index is $\kappa=\lfloor k/2\rfloor$; its summand has degree
$\kappa$ in $n$ and a nonzero leading coefficient, while every earlier
summand has smaller degree.  Hence Eq.~\eqref{S:fixedlevel-growth} gives the
stated values of $d_X(k)$.  When $k=1$ and hence $\lfloor k/2\rfloor=0$, the exact orthogonal counts are
$r(B_n,1)=3$ and $r(D_n,1)=4$.

We next bound the total quantum dimension, case by case.  For type $A$ at
level one, $\cD(SU(n)_1)=\sqrt n$.  For fixed $k\geq2$, the exact identity
$\cD_A(n,k)=\sqrt{n/k}\,\cD_A(k,n)
 \label{S:fixedlevel-A-duality}$
turns the growing-rank theory into the fixed Lie algebra $A_{k-1}$ at growing
level $n$; Lemma~\ref{lem:fixed-rank} then gives
$\log\cD_A(k,n)=O(\log n)$, and the square-root factor contributes only
$\tfrac12\log(n/k)$.  Similarly,
$\cD_C(n,k)=\cD_C(k,n)$ reduces every fixed-level type-$C$ sequence to the
fixed algebra $C_k$ at growing level.  For the orthogonal families with
$k\geq5$,
$\cD_{\mathfrak{so}_N}(k)=\cD_{\mathfrak{so}_k}(N)
 \label{S:fixedlevel-orthogonal-duality}$
does the same.  The remaining orthogonal levels $1\leq k\leq4$ are covered
directly by Lemma~\ref{lem:low-orthogonal}; level one has bounded sector count,
and at levels two through four it gives
$\log\cD=O(\log r)=O(\log\ell)$.  Since $n\asymp\ell$ in every classical
family, this proves both estimates in Eq.~\eqref{S:fixedlevel-growth}, and
therefore Eq.~\eqref{S:fixedlevelzero}.
\end{proof}

\subsection{Both parameters grow with a vanishing or divergent ratio}
\label{sec:unbalanced-regime}

\begin{lemma}[Large-level classical estimate]
\label{lem:large-level}
There is an absolute constant $C_0$ such that, for every classical simple Lie
algebra $\mathfrak g$ of Lie rank $\ell\geq2$ and every affine level $k$ with
$k/\ell\geq2$,
\begin{equation}
 \log\cD(\mathfrak g,k)\leq
 C_0\ell^2[1+\log(k/\ell)]+C_0\ell\log(k+\ell).
 \label{S:largelevel}
\end{equation}
\end{lemma}

\begin{proof}
For the four rank-$\ell$ root systems, the largest value of
$(\rho,\alpha)$ is respectively $\ell$, $2\ell-2$, $\ell$, and $2\ell-3$ for
$A_\ell,B_\ell,C_\ell,D_\ell$.  Since $k\geq2\ell$ and
$h^\vee=\ell+1,2\ell-1,\ell+1,2\ell-2$, respectively, every sine argument
$(\rho,\alpha)/(k+h^\vee)$ is strictly below $1/2$.  We may therefore use
$\sin(\pi x)\geq2x$ directly.  The values $2(\rho,\alpha)$ occupy $O(\ell)$
consecutive integer or half-integer sites, each with multiplicity at most
 $C_0\ell$.  The root contribution is then bounded by
 $C_0\ell\sum_{j\leq C_0\ell}\log[(k+C_0\ell)/j]$, which gives the first term in
Eq.~\eqref{S:largelevel}; the Kac--Peterson prefactors give the second.  At
this point the constant can be enlarged once to cover all four classical
families and the finitely many ranks suppressed by the uniform height
description.
\end{proof}

\begin{proposition}[Unbalanced two-parameter classical limits]
\label{prop:boundary}
Let $X=A,B,C,D$ and let $(n_j,k_j)$ be a sequence with
$n_j,k_j\to\infty$.  Write $\ell_j=n_j-1$ for $X=A$ and $\ell_j=n_j$ for
$X=B,C,D$.  If $k_j/\ell_j\to0$ or $k_j/\ell_j\to\infty$, then
\begin{equation}
 \frac{\log\cD_X(n_j,k_j)}{[\log r_X(n_j,k_j)]^2}\longrightarrow0.
 \label{S:boundaryzero}
\end{equation}
\end{proposition}

\begin{proof}
We first treat the case $k/\ell\to\infty$.  For integers
$1\leq a\leq b$, we use the binomial coefficient bounds
\begin{equation}
 \left(1+\frac ba\right)^a
 \leq \binom{a+b}{a}
 \leq \left[e\left(1+\frac ba\right)\right]^a.
 \label{S:binomial-boundary-bounds}
\end{equation}
Taking logarithms and using $\log(1+b/a) \geq \log 2$ gives
\begin{equation}
 a\log\left(1+\frac ba\right)
 \leq\log\binom{a+b}{a}
 \leq \left(1+\frac1{\log 2}\right)a\log\left(1+\frac ba\right).
 \label{S:log-binomial-boundary-bounds}
\end{equation}

For $A$ and $C$, substitute $(a,b)=(n-1,k)$ and $(n,k)$, respectively,
using Eq.~\eqref{S:rAC}.  Since $n\asymp\ell$ and
$k/\ell\to\infty$, Eq.~\eqref{S:log-binomial-boundary-bounds} gives
$\log r_A,\log r_C\asymp\ell\log(k/\ell)$.  For $B$ and $D$, put
$\kappa=\lfloor k/2\rfloor$ and use $(a,b)=(n-2,\kappa)$ or
$(n-3,\kappa)$ in the lower binomials of Eq.~\eqref{S:rBD}.  Its upper
bounds add only $2\log(k+2)$ or $3\log(k+3)$.  These additions are lower
order because
\begin{equation}
 \frac{\log k}{\ell\log(k/\ell)}
 =\frac{\log\ell}{\ell\log(k/\ell)}+\frac1\ell
 \longrightarrow0.
 \label{S:rankboundary-polynomial-factor}
\end{equation}
Thus, for every classical family,
\begin{equation}
 \log r_X(n,k)\asymp\ell\log(k/\ell),
 \qquad k/\ell\to\infty.
 \label{S:rankboundary}
\end{equation}

After discarding finitely many terms, $k/\ell\geq2$.  Therefore
Lemma~\ref{lem:large-level} and Eq.~\eqref{S:rankboundary} give
\begin{equation}
 \frac{\log\cD_X}{(\log r_X)^2}
 =O\!\left(\frac1{\log(k/\ell)}\right)
 +O\!\left(
 \frac{\log(k+\ell)}{\ell\log^2(k/\ell)}
 \right)\longrightarrow0.
 \label{S:boundary-large-level-ratio}
\end{equation}
For the second term, use
$\log(k+\ell)=\log\ell+\log(1+k/\ell)$.

It remains to consider $\ell/k\to\infty$.  The same sector-count calculation,
now using $k$ (or $\kappa\asymp k$ in types $B,D$) as the smaller binomial
parameter, gives
\begin{equation}
 \log r_X(n,k)\asymp k\log(\ell/k).
 \label{S:rankboundary-reversed}
\end{equation}
Apply the rank--level identities in Eq.~\eqref{S:duality}.  The dual theory
has Lie rank of order $k$ and level of order $\ell$, so it is already proved in the large-level case.  Types $C$ and $B,D$ preserve $\cD$ exactly.
In type $A$, duality contributes the additional term
$\tfrac12\log(n/k)=O(\log(\ell/k))$, which is negligible compared with
$[k\log(\ell/k)]^2$.  Apply Lemma~\ref{lem:large-level} to the dual
parameters and use Eq.~\eqref{S:rankboundary-reversed} for the denominator;
the same estimate as in Eq.~\eqref{S:boundary-large-level-ratio} then proves
Eq.~\eqref{S:boundaryzero} when $\ell/k\to\infty$.
\end{proof}

\subsection{From parameter regimes to the global cutoff envelope}
\label{sec:envelope}

\begin{lemma}[Finiteness at fixed cutoff]
\label{lem:fixed-cutoff}
For every fixed integer $R$,
\begin{equation}
 \sup\{\cD(\mathfrak g,k):r(\mathfrak g,k)\leq R\}<\infty.
 \label{S:boundedcutoff}
\end{equation}
\end{lemma}

\begin{proof}
For types $A,C$, the binomial formulas in Eq.~\eqref{S:rAC} bound both the
Lie-size parameter and the level.  For $B,D$ with $k\geq2$, the lower bounds
in Eq.~\eqref{S:rBD} first bound the Lie rank.  Setting all comark-two labels
to zero gives the additional bounds
\begin{equation}
 r_B(n,k)\geq\binom{k+2}{2},
 \qquad r_D(n,k)\geq\binom{k+3}{3},
\end{equation}
which then bound $k$.  The only unbounded-rank cases left are level one:
\begin{equation}
 r(B_n,1)=3,\qquad r(D_n,1)=4,\qquad
 \cD(B_n,1)=\cD(D_n,1)=2.
\end{equation}
For each of the five exceptional algebras, the sector count tends to infinity
with $k$.  These observations prove Eq.~\eqref{S:boundedcutoff}.
\end{proof}

\begin{proposition}[Sharp sequential bound]
\label{prop:sequential-upper}
For any sequence of simple WZW categories $\cC_j$, set
$r_j:=r(\cC_j)$ and $\cD_j:=\cD(\cC_j)$.  If $r_j\to\infty$, then
\begin{equation}
 \limsup_{j\to\infty}
 \frac{2\log\cD_j}{(\log_2r_j)^2}
 \leq\frac{7\zeta(3)}{4\pi^2}.
 \label{S:sequential}
\end{equation}
The bound is sharp: for the balanced symplectic sequence
$\cC_n=\cC(C_n,n)$, the limit in Eq.~\eqref{S:sequential} equals
$7\zeta(3)/(4\pi^2)$.
\end{proposition}

\begin{proof}
Choose a subsequence on which the left-hand side of
Eq.~\eqref{S:sequential} converges to its limsup.  Since there are only
finitely many Lie families, one family occurs infinitely often; retain those
terms.  This does not change the limit.

For an exceptional family the Lie rank is fixed, so
Lemma~\ref{lem:fixed-rank} gives the limit zero.  For a classical family,
write the parameters as $(n,k)$ and retain terms on which each integer
parameter is either fixed or tends to infinity.  Fixed $n$ is covered by
Lemma~\ref{lem:fixed-rank}, and fixed $k$ by
Lemma~\ref{lem:fixed-level}.  If both parameters tend to infinity, retain
terms on which the effective ratio $t$ converges in $[0,\infty]$.  A limit in
$(0,\infty)$ is bounded by Eq.~\eqref{eq:familycoefficients}, while the two
endpoint limits give zero by Proposition~\ref{prop:boundary}.  These cases
exhaust all possibilities and prove the upper bound.

For the balanced symplectic sequence,
\begin{equation}
 r_n=\binom{2n}{n},
 \qquad
 \log r_n=2n\log 2-\frac12\log(\pi n)+O(n^{-1}),
\end{equation}
while the proportional asymptotics give
\begin{equation}
 2\log\cD(C_n,n)=\frac{7\zeta(3)}{\pi^2}n^2+O(n\log n).
\end{equation}
Their ratio converges to $7\zeta(3)/(4\pi^2)$, proving sharpness.
\end{proof}

\begin{proposition}[Passage to the global cutoff]
\label{prop:cutoff-passage}
The sharp sequential bound, together with Lemma~\ref{lem:fixed-cutoff}, gives
\begin{equation}
 \lim_{R\to\infty}
 \frac{F_{\WZW}(R)}{(\log_2R)^2}
 =\frac{7\zeta(3)}{4\pi^2}.
 \label{S:envelope}
\end{equation}
\end{proposition}

\begin{proof}
Write
\begin{equation}
 c_*:=\frac{7\zeta(3)}{4\pi^2}.
\end{equation}

\emph{Lower bound.}
Let $r_n=r(C_n,n)=\binom{2n}{n}$.  For each large $R$, choose the largest
integer $n=n(R)$ for which $r_n\leq R$.  Then $n(R)\to\infty$ and
\begin{equation}
 r_n\leq R<r_{n+1},
 \qquad
 \frac{r_{n+1}}{r_n}=\frac{2(2n+1)}{n+1}<4.
\end{equation}
Thus $1\leq R/r_n<4$, and hence
\begin{equation}
 \frac{\log_2r_n}{\log_2R}\longrightarrow1.
 \label{S:saturating-log-gap}
\end{equation}
The category $\cC(C_n,n)$ is allowed at cutoff $R$, so
\begin{equation}
 \frac{F_{\WZW}(R)}{(\log_2R)^2}
 \geq
 \frac{2\log\cD(C_n,n)}{(\log_2r_n)^2}
 \left(\frac{\log_2r_n}{\log_2R}\right)^2.
\end{equation}
Both factors on the right have limits: the first tends to $c_*$ by
Proposition~\ref{prop:sequential-upper}, and the second tends to one by
Eq.~\eqref{S:saturating-log-gap}.  Therefore
\begin{equation}
 \liminf_{R\to\infty}
 \frac{F_{\WZW}(R)}{(\log_2R)^2}
 \geq c_*.
 \label{eq:lower-envelope}
\end{equation}

\emph{Upper bound.}
Consider any sequence of cutoffs $R_j\to\infty$.  At each cutoff, choose an
allowed category $\cC_j$, with $r_j=r(\cC_j)$ and $\cD_j=\cD(\cC_j)$,
satisfying
\begin{equation}
 r_j\leq R_j,
 \qquad
 2\log\cD_j\geq F_{\WZW}(R_j)-1.
 \label{S:cutoff-near-maximizer}
\end{equation}
The lower bound already proved shows that $F_{\WZW}(R_j)\to\infty$, so
Eq.~\eqref{S:cutoff-near-maximizer} gives $\cD_j\to\infty$.  The ranks must
therefore also tend to infinity.  Otherwise, infinitely many of the categories
would lie below one fixed rank cutoff, where Lemma~\ref{lem:fixed-cutoff}
bounds their total quantum dimensions.

We may now apply Proposition~\ref{prop:sequential-upper} to the sequence
$\cC_j$:
\begin{equation}
 \frac{2\log\cD_j}{(\log_2r_j)^2}\leq c_*+o(1).
\end{equation}
Using Eq.~\eqref{S:cutoff-near-maximizer} and $r_j\leq R_j$, we obtain
\begin{equation}
 \frac{F_{\WZW}(R_j)}{(\log_2R_j)^2}
 \leq\frac{1}{(\log_2R_j)^2}
 +(c_*+o(1))\left(\frac{\log_2r_j}{\log_2R_j}\right)^2
 \leq c_*+o(1).
\end{equation}
Since the cutoff sequence $R_j\to\infty$ was arbitrary, this proves the
matching upper bound.  Together with the lower bound, it proves
Eq.~\eqref{S:envelope}.
\end{proof}

Proposition~\ref{prop:cutoff-passage} proves Theorem~\ref{thm:main} directly
in its stated normalization.
For the selected two-circle torus bipartition, Eq.~\eqref{eq:state-maximization}
shows that the same expression is the maximal torus TEE;
the vacuum-flux state attains it.

\subsection{Finite stacking}
\begin{lemma}[Finite stacking bound]
\label{lem:stacking}
For every $\varepsilon>0$, there is a constant $C_\varepsilon$ such that every
finite product $\cC=\cC_1\boxtimes\cdots\boxtimes\cC_m$ in
$\mathfrak W^{\boxtimes}$ satisfies
\begin{equation}
 2\log\cD(\cC)\leq
 \left(\frac{7\zeta(3)}{4\pi^2}+\varepsilon\right)
 [\log_2r(\cC)]^2+C_\varepsilon\log_2r(\cC).
 \label{eq:stacking-bound}
\end{equation}
\end{lemma}

\begin{proof}
For finite semisimple Deligne products, the simple objects are tuples and
their quantum dimensions multiply \cite{Turaev}.  Thus, writing
$r_j:=r(\cC_j)$, $\cD_j:=\cD(\cC_j)$, and $q_j:=\log_2r_j$, the product
identities are
\begin{equation}
 r(\boxtimes_j\mathcal C_j)=\prod_jr_j,
 \qquad
 \cD(\boxtimes_j\mathcal C_j)=\prod_j\cD_j.
 \label{eq:stack-products}
\end{equation}
For every $\varepsilon>0$, choose $R_\varepsilon$ so that factors with
$r_j\geq R_\varepsilon$ obey
$2\log\cD_j\leq[7\zeta(3)/(4\pi^2)+\varepsilon]q_j^2$.
Equation~\eqref{S:boundedcutoff} supplies a constant $C_\varepsilon$ such that
the remaining nontrivial factors, $2\leq r_j<R_\varepsilon$, obey
$2\log\cD_j\leq C_\varepsilon q_j$.  Rank-one factors have $\cD_j=1$ and
do not contribute.  Hence
\begin{equation}
 2\log\cD(\boxtimes_j\mathcal C_j)\leq
 \left(\frac{7\zeta(3)}{4\pi^2}+\varepsilon\right)
 \left(\sum_jq_j\right)^2+C_\varepsilon\sum_jq_j.
\end{equation}
The sums are $\log_2r(\cC)$ by Eq.~\eqref{eq:stack-products}, proving
Eq.~\eqref{eq:stacking-bound}.
\end{proof}

\begin{corollary}[Semisimple WZW and finite stacking extension]
\label{cor:stacking}
For the semisimple, equivalently finitely stacked, envelope in
Eq.~\eqref{eq:stacked-envelope},
\begin{equation}
 \lim_{R\to\infty}
 \frac{F_{\WZW}^{\boxtimes}(R)}{(\log_2R)^2}
 =\frac{7\zeta(3)}{4\pi^2}.
 \label{eq:stacking-corollary}
\end{equation}
\end{corollary}

\begin{proof}
For categories with $r(\cC)\leq R$, Lemma~\ref{lem:stacking} gives
\begin{equation}
 \limsup_{R\to\infty}
 \frac{F_{\WZW}^{\boxtimes}(R)}{(\log_2R)^2}
 \leq\frac{7\zeta(3)}{4\pi^2}.
\end{equation}
The one-factor inclusion gives
$F_{\WZW}^{\boxtimes}(R)\geq F_{\WZW}(R)$ and hence the matching lower bound
from Theorem~\ref{thm:main}, proving
Eq.~\eqref{eq:stacking-corollary}.  Thus allowing any finite number of simple
summands, even a number depending on $R$, does not improve the leading
coefficient.
\end{proof}

\section{Conclusion and Outlook}

We formulated the unrestricted TQFT comparison through $Z(T^3)$ and $Z(S^3)$
in Eq.~\eqref{eq:general-envelope}, but solved only its simply connected WZW
restriction, both for individual simple theories and for finite uncoupled
Deligne stacks.  Theorem~\ref{thm:main} and
Corollary~\ref{cor:stacking} give, respectively,
\begin{equation}
 \lim_{R\to\infty}\frac{F_{\WZW}(R)}{(\log_2R)^2}
 =\lim_{R\to\infty}\frac{F_{\WZW}^{\boxtimes}(R)}{(\log_2R)^2}
 =\frac{7\zeta(3)}{4\pi^2}.
\end{equation}
The sequence $Sp(2n)_n$ supplies the transparent lower-bound construction.
The upper bound exhausts proportional and unbalanced two-parameter growth,
fixed Lie rank (including the exceptional algebras), and fixed level
(including the low orthogonal cases), before applying the cutoff-passage
lemma.  Balanced orthogonal sequences tie the leading coefficient, but no
unique maximizing theory at finite $R$ has been selected.
Lemma~\ref{lem:stacking} is the uniform estimate showing why the semisimple,
or finite-stacking, extension cannot raise the leading coefficient.

The leading tie leaves room for differences at finite cutoff.  Along the
exactly balanced sequences $C_n$ at $k=n$, $B_n$ at $k=2n+1$, and $D_n$ at
$k=2n$, Appendix~\ref{app:wzw-asymptotics} gives
$\log r_X=2n\log2+O(\log n)$ and a remainder $O(n\log n)$ in
$2\log\cD_X$.  Eliminating $n$ therefore gives, along each of these sequences,
\begin{equation}
 2\log\cD_X=
 \frac{7\zeta(3)}{4\pi^2}(\log_2r_X)^2
 +O(\log r_X\,\log\log r_X).
 \label{eq:balanced-subleading}
\end{equation}
This is an error bound, not a computed subleading coefficient or a remainder
estimate for the full envelope.  It does not order the families at finite
$R$.  Such a comparison requires sharper vacuum-product expansions and the
exact sector counts in Appendix~\ref{app:wzw-counts}, evaluated at the same
cutoff, including the discrete choices of Lie rank and level.  Determining
the maximizing theory at each finite $R$ remains open.

For the chosen torus bipartition, Section~\ref{sec:setup} proves that the
maximal TQFT contribution to the TEE over all torus ground states is
$2\log\cD$.  It is attained by the vacuum-flux state, and more generally by
every definite Abelian-flux state.  In our positive-magnitude convention the
constant contribution to $S_A$ is $-\Gamma_{T^2}$, and the factor two comes
from the two entangling circles.  Consequently, $F(R)$ is the supremum of these statewise
maxima over all admissible phases with $r(\cC)\leq R$, whereas $F_{\WZW}(R)$
restricts that outer supremum to simple WZW phases.  The categorical
theorem neither counts operational topological qubits nor addresses the
microscopic realizability of arbitrary $Sp(2n)_n$ phases.

The WZW value is a lower bound on the unrestricted envelope $F(R)$, not a
candidate proved to be universal.  It remains open whether other modular
categories, cosets, orbifolds, Dijkgraaf--Witten and gauged theories
\cite{DijkgraafWitten},
non-simply-connected Chern--Simons theories, or fermionic extensions produce
a larger leading envelope.  Subleading terms may also distinguish the
asymptotically tied $B,C,D$ families.

\acknowledgments

The author is supported by BIMSA and the NSFC under Grant No.~12505090.
The author thanks Xiao-Gang Wen and Ling-Yan Hung for discussions that
motivated this work.  Generative AI was used to assist with manuscript
preparation.  The author reviewed the resulting text and takes responsibility
for the final manuscript.

\bibliographystyle{JHEP}
\bibliography{references}

\clearpage
\appendix
\section{Modular data and Lie-theoretic input}
\label{app:wzw-data}

\subsection{Lie algebras, roots, and lattices}

Let \(\mathfrak g\) be a finite-dimensional complex simple Lie algebra.  A
Cartan subalgebra \(\mathfrak h\subset\mathfrak g\) is a maximal abelian
subalgebra consisting of semisimple elements, and the Lie-algebra rank is
\(\ell=\dim\mathfrak h\).  Simultaneous diagonalization of the adjoint action
of \(\mathfrak h\) gives the root-space decomposition
\begin{equation}
 \mathfrak g=\mathfrak h\oplus
 \bigoplus_{\alpha\in\Delta}\mathfrak g_\alpha,
 \qquad
 \mathfrak g_\alpha=
 \{X\in\mathfrak g:[H,X]=\alpha(H)X\ \text{for all }H\in\mathfrak h\}.
 \label{S:root-decomposition}
\end{equation}
The nonzero linear functionals \(\alpha\in\mathfrak h^*\) for which
\(\mathfrak g_\alpha\neq0\) are the roots.  A choice of simple roots
\(\Pi=\{\alpha_1,\ldots,\alpha_\ell\}\) divides the root system into positive
and negative roots,
\(\Delta=\Delta_+\sqcup(-\Delta_+)\); every positive root is a nonnegative
integer combination of the simple roots.
Our finite root-system conventions, including the classical realizations and
exceptional tables used below, follow the standard normalization in
Ref.~\cite{BourbakiLie}; the affine-algebra and integrability conventions
follow Ref.~\cite{KacBook}.

On the real span of the roots we use the Weyl-invariant inner product
\((\ ,\ )\) normalized so that every long root has squared length two.  The
coroot and the reflection associated with a root are
\begin{equation}
 \alpha^\vee=\frac{2\alpha}{(\alpha,\alpha)},
 \qquad
 s_\alpha(x)=x-(x,\alpha^\vee)\alpha.
 \label{S:coroot-reflection}
\end{equation}
The Weyl group \(W\) is generated by these reflections.  The fundamental
weights \(\Lambda_i\) are dual to the simple coroots, so the weight and
coroot lattices and the dominant cone are
\begin{align}
 (\Lambda_i,\alpha_j^\vee)&=\delta_{ij},&
 P&=\bigoplus_{i=1}^{\ell}\mathbb Z\Lambda_i,&
 P_+&=\bigoplus_{i=1}^{\ell}\mathbb Z_{\geq0}\Lambda_i,
 \nonumber\\
 Q^\vee&=\bigoplus_{i=1}^{\ell}\mathbb Z\alpha_i^\vee.&
 \label{S:root-lattices}
\end{align}
The Weyl vector and highest root are denoted by
\begin{equation}
 \rho=\frac12\sum_{\alpha\in\Delta_+}\alpha
      =\sum_{i=1}^{\ell}\Lambda_i,
 \qquad \theta=\text{the highest root}.
 \label{S:weyl-vector}
\end{equation}
The second expression for \(\rho\) follows from
\((\rho,\alpha_i^\vee)=1\).  The highest root is long; hence
\(\theta^\vee=\theta\) in our normalization.  The comarks \(a_i^\vee\) and
the dual Coxeter number are characterized by
\begin{equation}
 \theta=\sum_{i=1}^{\ell}a_i^\vee\alpha_i^\vee,
 \qquad
 h^\vee=1+(\rho,\theta)=1+\sum_{i=1}^{\ell}a_i^\vee.
 \label{S:comarks-dual-coxeter}
\end{equation}
Thus, if \(\lambda=\sum_i\lambda_i\Lambda_i\in P_+\), then
\((\lambda,\theta)=\sum_i a_i^\vee\lambda_i\).  This is the finite-dimensional
form of the affine integrability condition.

\subsection{Integrable weights and modular data}

Fix a positive integer level \(k\) and set \(K=k+h^\vee\).  The simple
objects of the untwisted WZW modular tensor category
\(\cC(\mathfrak g,k)\) are labeled by the dominant weights in the level
alcove
\begin{equation}
 P_k^+(\mathfrak g)=
 \left\{\lambda=\sum_{i=1}^{\ell}\lambda_i\Lambda_i:
 \lambda_i\in\mathbb Z_{\geq0},\quad
 \sum_i a_i^\vee\lambda_i\leq k\right\},
 \qquad r(\mathfrak g,k)=|P_k^+(\mathfrak g)|,
 \label{S:alcove}
\end{equation}
where \(a_i^\vee\) are the comarks.  The tensor unit is \(0\), and
\(\lambda^*=-w_0\lambda\), with \(w_0\) the longest Weyl-group element.
The lattice index in the modular normalization can be computed directly from
the coroot Gram matrix:
\begin{equation}
 \iota_{\mathfrak g}=|P/Q^\vee|
 =\det\bigl((\alpha_i^\vee,\alpha_j^\vee)_{i,j=1}^{\ell}\bigr).
 \label{S:lattice-index}
\end{equation}
Indeed, the entries in the $j$-th column of this Gram matrix are the
coordinates of $\alpha_j^\vee$ in the fundamental-weight basis.  Notice that
this is the coroot-lattice quotient $P/Q^\vee$, not the root-lattice quotient
$P/Q$; the distinction matters for non-simply-laced families.

With a common phase \(\varkappa_{\mathfrak g}\) chosen so that \(S_{00}>0\),
the Kac--Peterson matrix is
\begin{equation}
 S_{\lambda\mu}=
 \frac{\varkappa_{\mathfrak g}}
 {\sqrt{\iota_{\mathfrak g}}K^{\ell/2}}
 \sum_{w\in W}\epsilon(w)
 \exp\!\left[-\frac{2\pi\ii}{K}
 (w(\lambda+\rho),\mu+\rho)\right],
 \qquad |\varkappa_{\mathfrak g}|=1,\quad
 \iota_{\mathfrak g}=|P/Q^\vee|.
 \label{S:KPfull}
\end{equation}
This is the Kac--Peterson normalization \cite{KacPeterson}.
Here \(\epsilon(w)=\det(w)\) is the Weyl-group sign character.
The normalization in Eq.~\eqref{S:KPfull} is important for the optimization.
At \(\lambda=\mu=0\), the Weyl denominator identity gives the exact modulus
\begin{equation}
 \left|\sum_{w\in W}\epsilon(w)
 e^{-2\pi\ii(w\rho,\rho)/K}\right|
 =\prod_{\alpha\in\Delta_+}
 2\sin\left(\frac{\pi(\rho,\alpha)}{K}\right).
 \label{S:weyl-denominator-modulus}
\end{equation}
For every positive root,
\(0<(\rho,\alpha)\leq(\rho,\theta)=h^\vee-1<K\), so every sine on the
right-hand side is positive.  The phase \(\varkappa_{\mathfrak g}\) makes the
vacuum entry itself positive.  Consequently there are
no suppressed powers of two or parity factors: all lattice and power
normalizations outside the sine product are exactly
\(\iota_{\mathfrak g}^{-1/2}K^{-\ell/2}\).
The remaining data used below are
\begin{align}
 d_\lambda&=\frac{S_{0\lambda}}{S_{00}},&
 \cD&=\left(\sum_\lambda d_\lambda^2\right)^{1/2}=\frac1{S_{00}},
 \nonumber\\
 N_{\lambda\mu}^{\phantom{\lambda\mu}\nu}
 &=\sum_{\sigma\in P_k^+}
 \frac{S_{\lambda\sigma}S_{\mu\sigma}S_{\nu\sigma}^*}
 {S_{0\sigma}},&
 h_\lambda&=\frac{(\lambda,\lambda+2\rho)}{2K},
 \nonumber\\
 \theta_\lambda&=e^{2\pi\ii h_\lambda},&
 T_{\lambda\lambda}&=e^{2\pi\ii(h_\lambda-c/24)},
 \qquad c=\frac{k\dim\mathfrak g}{K}.
 \label{S:MTCdata}
\end{align}
These formulas specify the simple-object labels, fusion coefficients, quantum
dimensions, twists, and modular representation of the category
\cite{KacPeterson,Verlinde}.

The Weyl denominator reduces the vacuum entry to the positive root product
\begin{equation}
 S_{00}=\iota_{\mathfrak g}^{-1/2}K^{-\ell/2}
 \prod_{\alpha\in\Delta_+}
 2\sin\left(\frac{\pi(\rho,\alpha)}{K}\right).
 \label{S:S00}
\end{equation}
As a normalization check, take $A_1$ with its positive root $\alpha$.
Then $(\alpha,\alpha)=2$, $\Lambda_1=\rho=\alpha/2$, $h^\vee=2$, and
$P/Q^\vee=(\mathbb Z\alpha/2)/(\mathbb Z\alpha)$ has order two.  The
alcove consists of $\lambda=m\Lambda_1$, $0\leq m\leq k$, and
Eq.~\eqref{S:S00} becomes the standard $SU(2)_k$ result
\begin{equation}
 S_{00}=\sqrt{\frac{2}{k+2}}
 \sin\left(\frac{\pi}{k+2}\right).
 \label{S:SU2-check}
\end{equation}
The classical data are listed in Table~\ref{tab:classical-data}.
\begin{table}[t]
\centering
\caption{Classical WZW data in the long-root-length-two normalization.}
\label{tab:classical-data}
\begin{tabular}{c c c c c c c}
\toprule
type & compact group & \(\ell\) & \(h^\vee\) & \(\dim\mathfrak g\)
& comarks \(\{a_i^\vee\}\) & \(\iota_{\mathfrak g}\)\\
\midrule
\(A_{n-1}\) & \(SU(n)\) & \(n-1\) & \(n\) & \(n^2-1\)
& \(\{1^{n-1}\}\) & \(n\)\\
\(B_n\) & \(\Spin(2n+1)\) & \(n\) & \(2n-1\) & \(n(2n+1)\)
& \(\{1^2,2^{n-2}\}\) & \(4\)\\
\(C_n\) & \(Sp(2n)\) & \(n\) & \(n+1\) & \(n(2n+1)\)
& \(\{1^n\}\) & \(2^n\)\\
\(D_n\) & \(\Spin(2n)\) & \(n\) & \(2n-2\) & \(n(2n-1)\)
& \(\{1^3,2^{n-3}\}\) & \(4\)\\
\bottomrule
\end{tabular}
\end{table}
Here \(j^m\) means \(m\) occurrences of the comark \(j\).  The
identity \(h^\vee=1+\sum_i a_i^\vee\) and the determinant in
Eq.~\eqref{S:lattice-index} provide direct checks on the \(h^\vee\), comark,
and \(\iota_{\mathfrak g}\) columns.  The nonduplicated ranges used in
Theorem~\ref{thm:main} are \(A_{n-1}\) for
\(n\geq2\), \(B_n\) for \(n\geq3\), \(C_n\) for \(n\geq2\), and \(D_n\)
for \(n\geq4\).  The coincidences \(B_1=C_1=A_1\), \(B_2=C_2\), and
\(D_3=A_3\) are represented through their type-\(A\) or type-\(C\)
descriptions.  The semisimple case \(D_2=A_1\oplus A_1\) factorizes into its
two \(A_1\) summands and belongs to the finite-stacking envelope of
Corollary~\ref{cor:stacking}, not the simple-envelope theorem.  The
exceptional algebras have no variable rank; their data appear in
Table~\ref{tab:exceptional-data}.
\begin{table}[t]
\centering
\caption{Exceptional WZW data.}
\label{tab:exceptional-data}
\begin{tabular}{c c c c c c c}
\toprule
\(\mathfrak g\) & compact group & \(\ell\) & \(h^\vee\)
& \(\dim\mathfrak g\) & \(\{a_i^\vee\}\) & \(\iota_{\mathfrak g}\)\\
\midrule
\(G_2\) & \(G_2\) & 2 & 4 & 14 & \(\{1,2\}\) & 3\\
\(F_4\) & \(F_4\) & 4 & 9 & 52 & \(\{1,2^2,3\}\) & 4\\
\(E_6\) & \(E_6\) & 6 & 12 & 78 & \(\{1^2,2^3,3\}\) & 3\\
\(E_7\) & \(E_7\) & 7 & 18 & 133 & \(\{1,2^3,3^2,4\}\) & 2\\
\(E_8\) & \(E_8\) & 8 & 30 & 248 & \(\{2^2,3^2,4^2,5,6\}\) & 1\\
\bottomrule
\end{tabular}
\end{table}
For type \(C\), \((e_i,e_j)=\delta_{ij}/2\), so
\(P=\mathbb Z^n\) and \(Q^\vee=2\mathbb Z^n\).  This gives
\(\iota_{C_n}=2^n\), the normalization required by the exact rank--level
identity in Appendix~\ref{app:wzw-duality}.

\section{Sector counts}
\label{app:wzw-counts}

\subsection{Classical families}

For $X\in\{A,B,C,D\}$, let $r_X(n,k)$ and $\cD_X(n,k)$ denote the sector
count and total quantum dimension of the corresponding classical family in
Eq.~\eqref{eq:dictionary}.  This is the same family shorthand used in the main
text.  The external input is the standard affine alcove and comark data
\cite{KacBook}; all generating functions and coefficient estimates below are
then derived directly from Eq.~\eqref{S:alcove}.
We use the nonduplicated classical ranges \(n\geq2\) for \(A_{n-1}\) and
\(C_n\), \(n\geq3\) for \(B_n\), and \(n\geq4\) for \(D_n\), as in
Appendix~\ref{app:wzw-data}.

Introducing a slack variable in Eq.~\eqref{S:alcove} gives the generating
functions in Table~\ref{tab:sector-generating}.
\begin{table}[t]
\centering
\caption{Classical sector-count generating functions, where $r_X(n,k)$ is the
sector count in the displayed family.  The exponent of \((1-z)\) includes the
slack variable for unused level.}
\label{tab:sector-generating}
\begin{tabular}{c c c c}
\toprule
type & \(h^\vee\) & \(\sum_{k\geq0}r_X(n,k)z^k\) & proportional level \\
\midrule
\(A_{n-1}\) & \(n\) & \((1-z)^{-n}\) & \(k\sim tn\) \\
\(B_n\) & \(2n-1\) & \((1-z)^{-3}(1-z^2)^{-(n-2)}\)
 & \(k\sim2tn\) \\
\(C_n\) & \(n+1\) & \((1-z)^{-(n+1)}\) & \(k\sim tn\) \\
\(D_n\) & \(2n-2\) & \((1-z)^{-4}(1-z^2)^{-(n-3)}\)
 & \(k\sim2tn\) \\
\bottomrule
\end{tabular}
\end{table}
We now extract the four coefficients explicitly.  Write \([z^k]f(z)\) for
the coefficient of \(z^k\) in a formal power series \(f(z)\).  The
generalized binomial expansion is
\begin{equation}
 (1-x)^{-p}=\sum_{m\geq0}\binom{m+p-1}{p-1}x^m,
 \qquad p\in\mathbb Z_{>0}.
 \label{S:negative-binomial-series}
\end{equation}
Applying Eq.~\eqref{S:negative-binomial-series} directly to the type-\(A\)
and type-\(C\) generating functions gives
\begin{align}
 r_A(n,k)
 &=[z^k](1-z)^{-n}
  =[z^k]\sum_{m\geq0}\binom{m+n-1}{n-1}z^m
  =\binom{n+k-1}{n-1},\nonumber\\
 r_C(n,k)
 &=[z^k](1-z)^{-(n+1)}
  =[z^k]\sum_{m\geq0}\binom{m+n}{n}z^m
  =\binom{n+k}{n}.
 \label{S:rAC}
\end{align}

For types \(B\) and \(D\), compact parity formulas follow by first rewriting
the generating functions as
\begin{equation}
 \sum_{k\geq0}r_B(n,k)z^k=\frac{(1+z)^3}{(1-z^2)^{n+1}},
 \qquad
 \sum_{k\geq0}r_D(n,k)z^k=\frac{(1+z)^4}{(1-z^2)^{n+1}}.
 \label{S:BD-parity-generating}
\end{equation}
Let \(\kappa=\lfloor k/2\rfloor\) and
\(b_j=\binom{n+j}{n}\) for \(j\geq0\), with \(b_j=0\) for \(j<0\).
Equation~\eqref{S:negative-binomial-series} gives
\((1-z^2)^{-(n+1)}=\sum_{j\geq0}b_jz^{2j}\).
Separating the even and odd powers in the two numerators,
\begin{equation}
 (1+z)^3=(1+3z^2)+z(3+z^2),\qquad
 (1+z)^4=(1+6z^2+z^4)+4z(1+z^2),
\end{equation}
therefore yields
\begin{align}
 r_B(n,2\kappa)&=b_\kappa+3b_{\kappa-1},&
 r_B(n,2\kappa+1)&=3b_\kappa+b_{\kappa-1},\nonumber\\
 r_D(n,2\kappa)&=b_\kappa+6b_{\kappa-1}+b_{\kappa-2},&
 r_D(n,2\kappa+1)&=4(b_\kappa+b_{\kappa-1}).
 \label{S:BD-parity-coefficients}
\end{align}
To reduce these to one or two binomial coefficients, write
\(A_\kappa=\binom{n+\kappa-1}{n-1}\).  Then
\(b_\kappa=(n+\kappa)A_\kappa/n\) and
\(b_{\kappa-1}=\kappa A_\kappa/n\).
For the even type-\(D\) count, Pascal's identity also gives
\(b_\kappa-2b_{\kappa-1}+b_{\kappa-2}
=\binom{n+\kappa-2}{n-2}\).
Thus the exact counts are
\begin{align}
 r_B(n,2\kappa)&=\frac{n+4\kappa}{n}
                  \binom{n+\kappa-1}{n-1},&
 r_B(n,2\kappa+1)&=\frac{3n+4\kappa}{n}
                  \binom{n+\kappa-1}{n-1},\nonumber\\
 r_D(n,2\kappa)&=\frac{8\kappa}{n}
                  \binom{n+\kappa-1}{n-1}
                  +\binom{n+\kappa-2}{n-2},&
 r_D(n,2\kappa+1)&=\frac{4(n+2\kappa)}{n}
                  \binom{n+\kappa-1}{n-1}.
 \label{S:BD-compact-counts}
\end{align}
These formulas count the full simply connected categories, including spinor
weights.  In particular, \(r_B(n,1)=3\), \(r_D(n,1)=4\), and
\(r_D(n,2)=n+7\); at level one the type-\(D\) weights are
\(0,\Lambda_1,\Lambda_{n-1},\Lambda_n\), where \(\Lambda_i\) are the
fundamental weights defined in Appendix~\ref{app:wzw-data}.

For the estimates below it is also useful to keep the original convolution
form.  Expanding the two factors in Table~\ref{tab:sector-generating} by
Eq.~\eqref{S:negative-binomial-series}, the coefficient of \(z^k\) receives
contributions only from powers \(z^a z^{2j}\) with \(a+2j=k\).
Substituting \(a=k-2j\), so \(0\leq j\leq\kappa\), gives
\begin{align}
 r_B(n,k)&=\sum_{j=0}^{\kappa}
 \binom{k-2j+2}{2}\binom{j+n-3}{n-3},
 \label{S:rank-B-sum}\\
 r_D(n,k)&=\sum_{j=0}^{\kappa}
 \binom{k-2j+3}{3}\binom{j+n-4}{n-4}.
 \label{S:rank-D-sum}
\end{align}

For later use, these exact sums also give simple two-sided bounds.  The first
binomial in each summand is at least one, so the hockey-stick identity gives
the lower bounds.  It is also bounded above by \((k+2)^2\) in type \(B\) and
by \((k+3)^3\) in type \(D\).  Thus
\begin{align}
 \binom{\kappa+n-2}{n-2}
 &\leq r_B(n,k)\leq(k+2)^2\binom{\kappa+n-2}{n-2},\nonumber\\
 \binom{\kappa+n-3}{n-3}
 &\leq r_D(n,k)\leq(k+3)^3\binom{\kappa+n-3}{n-3}.
 \label{S:rBD}
\end{align}

\subsection{Fixed-rank asymptotics}

For any fixed simple algebra $\mathfrak g$, abbreviate
$r_{\mathfrak g}(k):=r(\mathfrak g,k)$ and
$\cD_{\mathfrak g}(k):=\cD(\mathfrak g,k)$.  Equation~\eqref{S:alcove} then
gives
\begin{equation}
 \sum_{k\geq0}r_{\mathfrak g}(k)z^k
 =\frac{1}{(1-z)\prod_{i=1}^{\ell}(1-z^{a_i^\vee})}.
 \label{S:fixedrankgenerating}
\end{equation}
The pole at \(z=1\), equivalently the volume of the comark simplex, yields
\begin{equation}
 r_{\mathfrak g}(k)
 =\frac{k^\ell}{\ell!\prod_i a_i^\vee}+O(k^{\ell-1}),
 \qquad
 \log r_{\mathfrak g}(k)=\ell\log k+O(1).
 \label{S:fixedrankcount}
\end{equation}
Expanding Eq.~\eqref{S:S00} at fixed root system gives
\begin{align}
 \cD_{\mathfrak g}(k)
 &=C_{\mathfrak g}K^{\ell/2+|\Delta_+|}
   \bigl[1+O(K^{-2})\bigr]\nonumber\\
 &=C_{\mathfrak g}K^{\dim\mathfrak g/2}
   \bigl[1+O(K^{-2})\bigr],
 \label{S:fixedrankdimension}
\end{align}
where
\begin{equation}
 C_{\mathfrak g}
 =\iota_{\mathfrak g}^{1/2}(2\pi)^{-|\Delta_+|}
  \prod_{\alpha\in\Delta_+}(\rho,\alpha)^{-1}.
\end{equation}
Therefore
\begin{equation}
  \log\cD_{\mathfrak g}(k)
  =\frac{\dim\mathfrak g}{2}\log k+O(1),
  \label{S:fixedranklogdimension}
\end{equation}
and $\log\cD_{\mathfrak g}(k)$ grows linearly with \(\log r_{\mathfrak g}(k)\), rather than quadratically, for every fixed-rank family.
For the five exceptional families the leading data are shown in
Table~\ref{tab:exceptional-asymptotics}.

Consequently every fixed-rank family has
\begin{equation}
 \frac{\log\cD_{\mathfrak g}(k)}{[\log r_{\mathfrak g}(k)]^2}
 =\frac{\dim\mathfrak g}{2\ell^2}\frac1{\log k}
  +O\!\left(\frac1{(\log k)^2}\right)\longrightarrow0.
 \label{S:exceptionaldecay}
\end{equation}
So the fixed-rank families do not contribute to the asymptotic growth of the total quantum dimension in Theorem~\ref{thm:main}.

\begin{table}[t]
\centering
\caption{Fixed-rank asymptotics for the exceptional families.}
\label{tab:exceptional-asymptotics}
\begin{tabular}{c c c c c}
\toprule
\(\mathfrak g\) & \(\ell\) & \(\dim\mathfrak g\)
& \(\log r\sim\) & \(\log\cD\sim\)\\
\midrule
\(G_2\) & 2 & 14 & \(2\log k\) & \(7\log k\)\\
\(F_4\) & 4 & 52 & \(4\log k\) & \(26\log k\)\\
\(E_6\) & 6 & 78 & \(6\log k\) & \(39\log k\)\\
\(E_7\) & 7 & 133 & \(7\log k\) & \(\tfrac{133}{2}\log k\)\\
\(E_8\) & 8 & 248 & \(8\log k\) & \(124\log k\)\\
\bottomrule
\end{tabular}
\end{table}
% Their coefficients of \(1/\log k\) in Eq.~\eqref{S:exceptionaldecay} are
% \(7/4,13/8,13/12,19/14\), and \(31/16\), respectively.  Thus
% \(\log\cD_{E_8}(k)\sim124\log k\) grows linearly with
% \(\log r_{E_8}(k)\sim8\log k\), rather than quadratically.

\section{Rank--level duality}
\label{app:wzw-duality}

\subsection{Meaning and scope of the duality}

The word ``rank'' in \emph{rank--level duality} refers to the size parameter
of the classical Lie algebra, not to the categorical rank \(r(\cC)\), which
counts simple objects.  The exchanges relevant here are schematically
\begin{equation}
 (A_{n-1},k)\longleftrightarrow(A_{k-1},n),\qquad
 (C_n,k)\longleftrightarrow(C_k,n),\qquad
 (\mathfrak{so}_N,k)\longleftrightarrow(\mathfrak{so}_k,N).
 \label{S:duality-exchanges}
\end{equation}
Thus the affine level on one side becomes the Lie-algebra size parameter on
the other.  This exchange is useful in the main proof because it relates a
large-level regime to the corresponding large-rank regime.

We need only a statement about total quantum dimension, not a categorical
equivalence.  Recall from Appendix~\ref{app:wzw-data} that
\(\cD=1/S_{00}\).  It is therefore enough to compare the vacuum entries of
the two modular matrices.  Write \(\cD_A(n,k)\) and \(\cD_C(n,k)\) for the
total quantum dimensions of types \(A_{n-1}\) and \(C_n\) at level \(k\).
For the orthogonal family, let
\begin{equation}
 \cC_N(k):=\cC(\mathfrak{so}_N,k),\qquad N\geq5,
\end{equation}
mean the full category of integrable modules of the untwisted affine algebra,
equivalently the simply connected \(\Spin(N)\) WZW modular data, including
spinor weights, and set
\(\cD_{\mathfrak{so}_N}(k):=\cD(\mathfrak{so}_N,k)\).
For \(N=5,6\), the coincidences \(\mathfrak{so}_5\simeq\mathfrak{sp}_4\) and
\(\mathfrak{so}_6\simeq\mathfrak{su}_4\) are understood with the same basic
inner product, normalized by long roots of squared length two.

With these conventions, the three identities used in the proof are
\begin{align}
 \cD_A(n,k)&=\sqrt{n/k}\,\cD_A(k,n),
 \qquad n,k\geq2,\nonumber\\
 \cD_C(n,k)&=\cD_C(k,n),
 \qquad n,k\geq1,\nonumber\\
 \cD_{\mathfrak{so}_N}(k)&=\cD_{\mathfrak{so}_k}(N),
 \qquad N,k\geq5.
 \label{S:duality}
\end{align}
These are specializations of established rank--level relations, not new
duality results.  For type \(A\), see the conformal-embedding analysis of
Altschuler, Bauer, and Itzykson~\cite{AltschulerBauerItzykson}; for type
\(C\), Verstegen~\cite{Verstegen} gives the equality of the modular
\(S\)-matrices under the rank--level exchange.  See also
Refs.~\cite{NaculichEtAl,NakanishiTsuchiya} for the general duality framework.
We retain short vacuum-normalization checks for types \(A,C\), followed by
the orthogonal parity and low-level details needed in the fixed-level
argument.

The orthogonal equality means precisely
\(S_{00}^{\cC_N(k)}=S_{00}^{\cC_k(N)}\).  It does not assert equality of
categorical ranks, a braided equivalence of the full categories, or an
unqualified matching of tensor and spinor sectors.  It also does not extend
automatically to \(SO(N)\) quotients, gauged theories, spin theories, or other
global forms of orthogonal Chern--Simons theory.  The checks below concern
only the displayed vacuum identities with our Kac--Peterson normalization.

\subsection{Vacuum normalization and complementary minors}

For a simple algebra of Lie rank \(\ell\), isolate the positive-root product
in Eq.~\eqref{S:S00} by writing
\begin{equation}
 P_{\mathfrak g}(M)
 :=\prod_{\alpha\in\Delta_+}
 2\sin\left(\frac{\pi(\rho,\alpha)}{M}\right),
 \qquad
 S_{00}=\iota_{\mathfrak g}^{-1/2}M^{-\ell/2}
 P_{\mathfrak g}(M).
 \label{S:duality-root-product}
\end{equation}
The shifted denominator \(M=k+h^\vee\) is unchanged by each exchange in
Eq.~\eqref{S:duality-exchanges}:
\begin{equation}
 M=n+k\quad(A),\qquad
 M=n+k+1\quad(C),\qquad
 M=N+k-2\quad(\mathfrak{so}).
 \label{S:duality-shifted-denominators}
\end{equation}
The problem is therefore reduced to comparing two finite sine products with
the same denominator.  In type \(A\) this can be done directly with the
complete sine product.  In types \(C,B,D\), the Weyl denominator writes the
root product as a determinant of a finite sine or cosine transform.

We will repeatedly use the following elementary form of Jacobi's
complementary-minor theorem.  If \(U\) is a real orthogonal \(d\times d\)
matrix and \(I,J\) are index sets of the same size, then
\begin{equation}
 \left|\det U_{I,J}\right|
 =\left|\det U_{I^c,J^c}\right|.
 \label{S:Jacobi-minors}
\end{equation}
It follows from the cofactor formula and \(U^{-1}=U^{\mathsf T}\).  In our
applications, the first determinant encodes the root product before the
rank--level exchange and the complementary determinant encodes the product
after the exchange.  Keeping the normalizing scalar of \(U\) is what fixes all
powers of \(M\) and two.

\subsection{Type A: a normalization check}

For \(A_{n-1}\), the positive roots are \(e_i-e_j\), \(i<j\).  The height
\((\rho,e_i-e_j)=j-i=m\) occurs \(n-m\) times.  Hence, with \(M=n+k\) and
\(s_m=2\sin(\pi m/M)\), the vacuum root product is
\begin{equation}
 P_A(n,k)=\prod_{m=1}^{n-1}s_m^{n-m},\qquad
 \cD_A(n,k)=\frac{\sqrt n\,M^{(n-1)/2}}{P_A(n,k)}.
\end{equation}
Reflection \(s_{M-m}=s_m\) writes the ratio of root products as
\(\prod_{m=1}^{M-1}s_m^{n-m}\).  Pairing reflected factors and using the
complete sine product \(\prod_{m=1}^{M-1}s_m=M\) gives
\begin{equation}
 \frac{P_A(n,k)}{P_A(k,n)}=M^{(n-k)/2},\qquad
 \frac{\cD_A(n,k)}{\cD_A(k,n)}=\sqrt{n/k}.
\end{equation}
This checks the vacuum specialization of the known type-\(A\) relation
\cite{AltschulerBauerItzykson,NaculichEtAl}.  The square-root factor in the
type-\(A\) identity comes from the lattice indices
\(\iota_{A_{n-1}}=n\) and \(\iota_{A_{k-1}}=k\).

\subsection{Type C: a normalization check}

For \(C_n\), the common shifted denominator is \(M=n+k+1\).  Set
\(a_j=2\sin[\pi j/(2M)]\).  Reversing the Weyl-vector coordinates, the vacuum
root product is
\begin{equation}
 P_C(n,k)=\prod_{p=1}^n a_{2p}
 \prod_{1\leq q<p\leq n}a_{p-q}a_{p+q}
 =\left|\det_{1\leq p,q\leq n}
 \left[2\sin\frac{\pi pq}{M}\right]\right|.
 \label{S:Calternant}
\end{equation}
The relevant finite transform is the \((M-1)\times(M-1)\) orthogonal sine
matrix
\begin{equation}
 U_{pq}=\sqrt{\frac2M}\sin\frac{\pi pq}{M},\qquad
 \left|\det U_{[1,n]}\right|
 =(2M)^{-n/2}P_C(n,k).
 \label{S:C-normalized-minor}
\end{equation}
Here \(1\leq p,q\leq M-1\), and \(U_{[a,b]}\) denotes the principal
submatrix whose indices are \(a,a+1,\ldots,b\).
By Eq.~\eqref{S:Jacobi-minors}, its absolute determinant equals that of the
complementary \(k\times k\) minor.
Reflecting the complementary indices by \(p\mapsto M-p\) turns the latter
into the leading \(k\times k\) minor and changes only row and column signs.
It follows that
\begin{equation}
 \frac{P_C(n,k)}{P_C(k,n)}=(2M)^{(n-k)/2},\qquad
 \cD_C(n,k)=\frac{(2M)^{n/2}}{P_C(n,k)}.
\end{equation}
The second formula uses \(\iota_{C_n}=2^n\).  Substitution checks
\(\cD_C(n,k)=\cD_C(k,n)\), the vacuum specialization of the modular
\(S\)-matrix relation in Ref.~\cite{Verstegen}.

\subsection{Orthogonal families}

For the orthogonal series, set
\(\ell_N=\lfloor N/2\rfloor\) and \(\ell_k=\lfloor k/2\rfloor\), and write
\(\mathfrak{so}_N\) as
\(B_{(N-1)/2}\) or \(D_{N/2}\) according to the parity of \(N\), and set
\(M=N+k-2\).  Thus odd \(N\) gives type \(B\), while even \(N\) gives type
\(D\).  Exchanging \(N\) and \(k\) can preserve the type or switch between
\(B\) and \(D\), which is why the parity cases must be kept explicit.  The
vacuum Weyl products are
\begin{align}
 P_{B_n}(M)&=\prod_{i=1}^n2\sin\frac{\pi(n-i+1/2)}M
 \prod_{i<j}2\sin\frac{\pi(j-i)}M
 2\sin\frac{\pi(2n-i-j+1)}M,\nonumber\\
 P_{D_n}(M)&=\prod_{i<j}2\sin\frac{\pi(j-i)}M
 2\sin\frac{\pi(2n-i-j)}M.
 \label{S:orthogonal-products}
\end{align}
Since \(\iota_{B_n}=\iota_{D_n}=4\), Eq.~\eqref{S:S00} reads in both parity
families
\begin{equation}
 S_{00}^{\cC_N(k)}
 =\frac12M^{-\lfloor N/2\rfloor/2}P_{\mathfrak{so}_N}(M),
 \qquad M=N+k-2.
 \label{S:orthogonal-S00}
\end{equation}
The Weyl denominator identities give
\begin{align}
 P_{B_n}(M)
 &=\left|\det_{0\leq a,b<n}
 \left[2\sin\frac{2\pi(a+1/2)(b+1/2)}M\right]\right|,
 \label{S:Bdet}\\
 P_{D_n}(M)
 &=2^{n-1}\left|\det_{0\leq a,b<n}
 \left[\cos\frac{2\pi ab}M\right]\right|.
 \label{S:Ddet}
\end{align}
The half-integer coordinates in Eq.~\eqref{S:Bdet} are the type-\(B\)
Weyl-vector coordinates, while the integer coordinates in
Eq.~\eqref{S:Ddet} are their type-\(D\) counterparts.  Absolute values remove
signs caused by reordering rows or columns.

There are three nonredundant parity cases: odd--odd, even--even, and
odd--even; the even--odd case is obtained by transposition.  We keep every
normalizing scalar in the corresponding real orthogonal transform.

\begin{center}
\small
\begin{tabular}{c c c}
\toprule
parities of \((N,k)\) & exchanged types & orthogonal transform\\
\midrule
odd--odd & \(B_{\ell_N}\leftrightarrow B_{\ell_k}\) & discrete sine, type IV\\
even--even & \(D_{\ell_N}\leftrightarrow D_{\ell_k}\) & discrete cosine, type I\\
odd--even & \(B_{\ell_N}\leftrightarrow D_{\ell_k}\) & odd-length cosine\\
\bottomrule
\end{tabular}
\end{center}

\paragraph{Odd--odd.}
If \(N=2\ell_N+1\) and \(k=2\ell_k+1\), put
\(L=\ell_N+\ell_k\), so \(M=2L\), and use the
\(L\times L\) discrete sine transform of type IV
\begin{equation}
 U^{oo}_{ab}=\sqrt{\frac4M}
 \sin\frac{2\pi(a+1/2)(b+1/2)}M,
 \qquad0\leq a,b<L.
\end{equation}
The absolute determinant of its leading \(\ell_N\times\ell_N\) minor is
\(M^{-\ell_N/2}P_{B_{\ell_N}}(M)\).  Reversing the indices of the complementary minor
changes only signs and identifies its absolute determinant with
\(M^{-\ell_k/2}P_{B_{\ell_k}}(M)\).

\paragraph{Even--even.}
If \(N=2\ell_N\) and \(k=2\ell_k\), put
\(L=\ell_N+\ell_k-1\), so \(M=2L\).  With
\(c_0=c_L=2^{-1/2}\) and \(c_a=1\) otherwise, use the discrete cosine
transform of type I
\begin{equation}
 U^{ee}_{ab}=\sqrt{\frac4M}\,c_a c_b
 \cos\frac{2\pi ab}M,
 \qquad0\leq a,b\leq L.
\end{equation}
The endpoint weights contribute \(1/2\) to either relevant determinant, and
Eq.~\eqref{S:Ddet} gives
\begin{equation}
 |\det U^{ee}_{[0,\ell_N-1]}|=M^{-\ell_N/2}P_{D_{\ell_N}}(M),
 \qquad
 |\det U^{ee}_{[\ell_N,L]}|=M^{-\ell_k/2}P_{D_{\ell_k}}(M).
\end{equation}
Index reversal identifies the second determinant with the leading type-\(D\)
minor up to signs.

\paragraph{Odd--even.}
Finally suppose \(N=2\ell_N+1\) and \(k=2\ell_k\); the even--odd case follows by
interchanging the two sides.  Now \(M=2(\ell_N+\ell_k)-1\),
\(L=(M-1)/2=\ell_N+\ell_k-1\), and the orthogonal odd-length cosine transform is
\begin{equation}
 U^{oe}_{ab}=\sqrt{\frac4M}\,c_a c_b
 \cos\frac{2\pi ab}M,
 \qquad0\leq a,b\leq L,
 \quad c_0=2^{-1/2},\quad c_a=1\ (a>0).
\end{equation}
Its leading \(\ell_k\times\ell_k\) minor has absolute determinant
\(M^{-\ell_k/2}P_{D_{\ell_k}}(M)\).  In the complementary minor set
\(\xi=M/2-a\) and \(\eta=M/2-b\).  Reversing the indices turns them into the
half-integers \(1/2,3/2,\ldots,\ell_N-1/2\), while
\begin{equation}
 \left|\cos\frac{2\pi ab}M\right|
 =\left|\sin\frac{2\pi\xi\eta}M\right|
\end{equation}
up to independently factorizable row and column signs.  Equation~\eqref{S:Bdet}
therefore identifies the complementary determinant with
 \(M^{-\ell_N/2}P_{B_{\ell_N}}(M)\).

\paragraph{Transform normalization.}
For completeness, the elementary orthogonality sums that fix all transform
normalizations are
\begin{align}
 &\sum_{b=0}^{L-1}
 \sin\frac{2\pi(a+1/2)(b+1/2)}M
 \sin\frac{2\pi(a'+1/2)(b+1/2)}M
 =\frac M4\delta_{aa'}
 &&\text{(odd--odd)},\nonumber\\
 &\sum_{b=0}^{L}c_b^2
 \cos\frac{2\pi ab}M\cos\frac{2\pi a'b}M
 =\frac{M}{4c_a^2}\delta_{aa'}
 &&\text{(even--even and mixed)}.
 \label{S:orthogonality-sums}
\end{align}
In the second line the even--even case has
\(c_0=c_L=2^{-1/2}\), whereas the odd-length mixed case has
\(c_0=2^{-1/2}\) and \(c_a=1\) for \(a>0\).  Product-to-sum and finite
geometric-series summation prove these identities.  They show directly that
no unrecorded determinant scalar is present.

Jacobi's theorem in these three cases proves, for all four parity choices,
\begin{equation}
 M^{-\lfloor N/2\rfloor/2}P_{\mathfrak{so}_N}(M)
 =M^{-\lfloor k/2\rfloor/2}P_{\mathfrak{so}_k}(M).
\end{equation}
No parity-dependent scalar remains because both types \(B,D\) have
\(|P/Q^\vee|=4\).  Multiplying by the common Kac--Peterson factor \(1/2\)
proves the orthogonal identity in Eq.~\eqref{S:duality}.

\subsection{Low orthogonal levels}

The levels \(1\leq k\leq4\) used in Eq.~\eqref{S:lowO} lie outside the stable
orthogonal range \(N,k\geq5\), so we evaluate them directly rather than
interpreting them as dualities with low-rank orthogonal algebras.  Set
\begin{equation}
 s_x=2\sin\frac{\pi x}{M},\qquad
 \mathcal I_M=\prod_{j=1}^{M-1}s_j=M,\qquad
 \mathcal H_M=\prod_{j=1}^{M}s_{j-1/2}=2,
 \qquad s_x=s_{M-x}.
 \label{S:completeOproducts}
\end{equation}
Here \(\mathcal I_M\) and \(\mathcal H_M\) are the complete integer and
half-integer sine products, respectively.  They are the only product
identities needed for the four direct evaluations.
\mbox{Let \(\ell=\lfloor N/2\rfloor\).}  Pairing the factors in
Eq.~\eqref{S:orthogonal-products} with their reflections and substituting
directly into Eq.~\eqref{S:orthogonal-S00} gives
Table~\ref{tab:low-orthogonal}.  The middle column is the raw Weyl product,
not a category-level duality involving a low-rank orthogonal algebra.
\begin{table}[t]
\centering
\caption{Vacuum products for the four low orthogonal levels.}
\label{tab:low-orthogonal}
\renewcommand{\arraystretch}{2}
\small
\begin{tabular}{c c c c c}
\toprule
\(k\) & \(M\) & raw \(P_{\mathfrak{so}_N}(M)\) & \(S_{00}\) & \(\cD\)\\
\midrule
1 & \(N-1\) & \(M^{\ell/2}\) & \(1/2\) & \(2\)\\
2 & \(N\) & \(M^{(\ell-1)/2}\) & \(1/(2\sqrt N)\) & \(2\sqrt N\)\\
3 & \(N+1\) & \(s_{1/2}M^{(\ell-1)/2}\)
& \(\dfrac{\sin[\pi/(2(N+1))]}{\sqrt{N+1}}\)
& \(\dfrac{\sqrt{N+1}}{\sin[\pi/(2(N+1))]}\)\\
4 & \(N+2\) & \(s_1^2M^{(\ell-2)/2}\)
& \(\dfrac{2\sin^2[\pi/(N+2)]}{N+2}\)
& \(\dfrac{N+2}{2\sin^2[\pi/(N+2)]}\)\\
\bottomrule
\end{tabular}
\end{table}
The associated exact sector counts, obtained from
Table~\ref{tab:sector-generating}, are
\begin{equation}
\begin{array}{c|cccc}
 &k=1&k=2&k=3&k=4\\ \hline
B_n&3&n+4&3n+4&(n^2+9n+8)/2\\
D_n&4&n+7&4n+8&(n^2+15n+16)/2.
\end{array}
\label{S:low-orthogonal-counts}
\end{equation}
Thus at level one both rank and \(\cD\) remain bounded; at levels two through
four, \(\log\cD=O(\log n)=O(\log r)\).  This directly verifies every
low-level formula used in Eq.~\eqref{S:lowO} against the standard affine
vacuum product.
% As a regression check, the ancillary program
% \path{anc/check_classical_asymptotics.py} evaluates the same
% Kac--Peterson products for every parity pair \(5\leq N,k\leq30\) and for all
% four low levels with \(5\leq N\leq60\).  This high-precision diagnostic is not
% used in place of the determinant and product proofs above.

\section{Proportional WZW asymptotics}
\label{app:wzw-asymptotics}

Define
\begin{align}
 H(t)&=(1+t)\log(1+t)-t\log t,\nonumber\\
 J(t)&=-\int_0^1(1-x)
 \log\left(2\sin\frac{\pi x}{1+t}\right)\dd x.
 \label{S:HJ}
\end{align}
We now derive, family by family, the two formulas used in the proportional
optimization.  Fix a closed interval \(I=[a,b]\) with
\(0<a\leq b<\infty\).  Throughout this appendix, every \(O(\,\cdot\,)\)
estimate is uniform for \(t\in I\): its implied constant may depend on the
fixed interval \(I\), but not on \(n\) or \(k\).  With this convention we can
use the ordinary \(O\) notation instead of carrying an interval subscript in
every formula.
The result is
\begin{align}
 \log r_A(n,k)&=nH(t_A)+O(\log n),&
 \log\cD_A(n,k)&=n^2J(t_A)+O(n\log n),\nonumber\\
 \log r_C(n,k)&=nH(t_C)+O(\log n),&
 \log\cD_C(n,k)&=2n^2J(t_C)+O(n\log n),\nonumber\\
 \log r_X(n,k)&=nH(t_X)+O(\log n),&
 \log\cD_X(n,k)&=2n^2J(t_X)+O(n\log n),
 \quad X=B,D,
 \label{S:asym}
\end{align}
where \(t_A=t_C=k/n\) and \(t_B=t_D=k/(2n)\).
The proof is divided into four explicit steps.

\paragraph{Step 1: sector counts and the effective variable.}
We begin from the logarithmic form of Stirling's formula
\begin{equation}
 \log m!
 =\left(m+\frac12\right)\log m-m+\frac12\log(2\pi)
 +O(m^{-1}).
 \label{S:log-Stirling}
\end{equation}
For this expansion and its remainder, see Ref.~\cite{NISTHandbook}.
To see exactly how \(H(t)\) arises, fix integers \(c_0,c_1\) and set
\[
 N=(1+t)n+c_0,\qquad M=n+c_1,\qquad P=N-M=tn+c_0-c_1.
\]
Because \(t\in[a,b]\), all three arguments \(N,M,P\) are bounded above and
below by positive multiples of \(n\).  Using
\(\binom NM=N!/(M!P!)\) and applying Eq.~\eqref{S:log-Stirling} three times
gives
\begin{align}
 \log\binom NM
 &=N\log N-M\log M-P\log P
 +\frac12\log\frac{N}{2\pi MP}+O(n^{-1}).
 \label{S:Stirling-subtraction}
\end{align}
The bounded shifts \(c_0,c_1\) contribute at most \(O(\log n)\), while
\begin{align}
 N\log N-M\log M-P\log P
 ={}&n\Bigl[(1+t)\bigl(\log n+\log(1+t)\bigr)
 -\log n\nonumber\\
 &\hspace{3.2cm}
 -t\bigl(\log n+\log t\bigr)\Bigr]+O(\log n) \nonumber\\
  ={}&n\Bigl[(1+t)\log(1+t)-t\log t\Bigr]+O(\log n).
 \label{S:Stirling-leading}
\end{align}
And the half-logarithmic term \(\frac12\log\frac{N}{2\pi MP}\) in Eq.~\eqref{S:Stirling-subtraction} is \(O(\log n)\).
We have therefore proved the uniform shifted-binomial formula
\begin{equation}
 \log\binom{(1+t)n+c_0}{n+c_1}
 =nH(t)+O(\log n).
 \label{S:uniform-Stirling}
\end{equation}

For \(A_{n-1}\), set \(t=k/n\).  The exact count in
Eq.~\eqref{S:rAC} becomes
\begin{equation}
 r_A(n,k)=\binom{n+k-1}{n-1}
 =\binom{(1+t)n-1}{n-1},
\end{equation}
so Eq.~\eqref{S:uniform-Stirling}, with \(c_0=c_1=-1\), gives
\begin{equation}
 \log r_A(n,k)=nH(k/n)+O(\log n).
\end{equation}
For \(C_n\), the same calculation has no bounded shifts:
\begin{equation}
 r_C(n,k)=\binom{n+k}{n}
 =\binom{(1+t)n}{n},
\end{equation}
and therefore
\begin{equation}
 \log r_C(n,k)=nH(k/n)+O(\log n).
 \label{S:rank-AC-derivation}
\end{equation}

For the orthogonal families, put \(\kappa=\lfloor k/2\rfloor\).  Appendix
\ref{app:wzw-counts} derives the exact coefficient sums
in Eqs.~\eqref{S:rank-B-sum} and \eqref{S:rank-D-sum}; their two-sided
consequences are recorded in Eq.~\eqref{S:rBD}.  We now use those count
formulas only to extract their proportional asymptotics.  The bounds show
that the only nonpolynomial part is the binomial coefficient with parameters
 \(n\) and \(\kappa\).

Now set \(t:=k/(2n)\).  Since \(\kappa=\lfloor k/2\rfloor\),
\begin{equation}
 \frac \kappa n=t+O(n^{-1}),
 \label{S:orthogonal-effective-variable}
\end{equation}
which is the origin of the factor two in the orthogonal effective variable.
 For all sufficiently large \(n\), both \(t\) and \(\kappa/n\) lie in a slightly
larger fixed compact subinterval of \((0,\infty)\).  On that interval the
derivative
\[
 H'(t)=\log\frac{1+t}{t}
\]
is bounded.  Hence the mean-value theorem gives
\[
 nH(\kappa/n)=nH(t)+O(1).
\]
Apply Eq.~\eqref{S:uniform-Stirling} to the binomials in Eq.~\eqref{S:rBD},
 using \(\kappa/n\) in place of \(t\).  The mean-value estimate above then yields
\begin{align}
 \log\binom{\kappa+n-2}{n-2}&=nH(t)+O(\log n),\nonumber\\
 \log\binom{\kappa+n-3}{n-3}&=nH(t)+O(\log n).
\end{align}
Finally, \(k=2tn=O(n)\), so
\(\log(k+2)^2,\log(k+3)^3=O(\log n)\).  Taking logarithms in
Eq.~\eqref{S:rBD} and squeezing gives
\begin{equation}
 \log r_B(n,k)=nH(k/(2n))+O(\log n),\qquad
 \log r_D(n,k)=nH(k/(2n))+O(\log n).
 \label{S:rank-BD-derivation}
\end{equation}
This completes the sector-count asymptotics for all four classical families.

\paragraph{Step 2: normalization and exact root-height sums.}
Taking the reciprocal and logarithm of the normalized Kac--Peterson product
in Eq.~\eqref{S:S00} gives the exact identity
\begin{align}
 \cD&=\iota_{\mathfrak g}^{1/2}K^{\ell/2}
 \prod_{\alpha\in\Delta_+}
 \left(2\sin\frac{\pi(\rho,\alpha)}K\right)^{-1},\nonumber\\
 \log\cD&=\frac12\log\iota_{\mathfrak g}
 +\frac\ell2\log K
 +\sum_{\alpha\in\Delta_+}L_K((\rho,\alpha)),
 \qquad
 L_K(h):=-\log\left(2\sin\frac{\pi h}{K}\right).
 \label{S:logD-decomposition}
\end{align}
Thus no power of two, lattice factor, or parity factor is hidden in the
continuum limit.  The relevant data and every positive-root height are
summarized in the table below.  To derive its last two columns, use the
positive roots
\[
 \begin{array}{c|c}
 A_{n-1}&e_i-e_j\\
 B_n&e_i-e_j,\ e_i+e_j,\ e_i\\
 C_n&e_i-e_j,\ e_i+e_j,\ 2e_i\\
 D_n&e_i-e_j,\ e_i+e_j
 \end{array}
 \qquad(1\leq i<j\leq n)
\]
and the Weyl vectors
\begin{align}
 \rho_A&=\sum_{i=1}^n\left(\frac{n+1}{2}-i\right)e_i,&
 \rho_B&=\sum_{i=1}^n\left(n-i+\frac12\right)e_i,\nonumber\\
 \rho_C&=\sum_{i=1}^n(n+1-i)e_i,&
 \rho_D&=\sum_{i=1}^n(n-i)e_i.
 \label{S:classical-Weyl-vectors}
\end{align}
These classical positive-root systems and Weyl vectors are standard; see
Ref.~\cite{BourbakiLie}.  They are displayed here because their individual
height multiplicities, rather than only their total numbers, control the
continuum limit.
For \(A,B,D\), \((e_i,e_j)=\delta_{ij}\) on the root span; for \(C\),
\((e_i,e_j)=\delta_{ij}/2\).  Taking the displayed inner products gives every
height in the table.

\begin{center}
\small
\emph{Exact data used in the proportional root sums
(\(1\leq i<j\leq n\)); the last two columns are \((\rho,\alpha)\).}
\smallskip

\begin{tabular}{c c c c c}
\toprule
family & \(\iota_{\mathfrak g}\) & \(K=k+h^\vee\) & pair-root heights & one-body heights\\
\midrule
\(A_{n-1}\) & \(n\) & \(k+n\)
 & \(j-i\) & none\\
\(B_n\) & \(4\) & \(k+2n-1\)
 & \(j-i,\ 2n+1-i-j\) & \(n-i+\tfrac12\)\\
\(C_n\) & \(2^n\) & \(k+n+1\)
 & \(\tfrac{j-i}{2},\ \tfrac{2n+2-i-j}{2}\) & \(n+1-i\)\\
\(D_n\) & \(4\) & \(k+2n-2\)
 & \(j-i,\ 2n-i-j\) & none\\
\bottomrule
\end{tabular}
\end{center}

For completeness, inserting the table into Eq.~\eqref{S:logD-decomposition}
gives the four finite sums before any limit is taken:
\begin{align}
 \log\cD_A={}&\frac12\log n+\frac{n-1}{2}\log(k+n)
 +\sum_{i<j}L_{k+n}(j-i),
 \label{S:finite-A}\\
 \log\cD_B={}&\log2+\frac n2\log(k+2n-1)
 +\sum_{i<j}\!\left[L_{k+2n-1}(j-i)
 +L_{k+2n-1}(2n+1-i-j)\right]\nonumber\\
 &+\sum_{i=1}^nL_{k+2n-1}(n-i+\tfrac12),
 \label{S:finite-B}\\
 \log\cD_C={}&\frac n2\log2+\frac n2\log(k+n+1)
 +\sum_{i<j}\!\left[L_{k+n+1}(\tfrac{j-i}{2})
 +L_{k+n+1}(\tfrac{2n+2-i-j}{2})\right]\nonumber\\
 &+\sum_{i=1}^nL_{k+n+1}(n+1-i),
 \label{S:finite-C}\\
 \log\cD_D={}&\log2+\frac n2\log(k+2n-2)
 +\sum_{i<j}\!\left[L_{k+2n-2}(j-i)
 +L_{k+2n-2}(2n-i-j)\right].
 \label{S:finite-D}
\end{align}
The lattice and power prefactors in these formulas are all
\(O(n\log n)\).  The one-body sums in types \(B,C\) contain only \(n\)
terms, each of size \(O(\log n)\), and are also \(O(n\log n)\).
Therefore only the pair-root sums can contribute at order \(n^2\).

\paragraph{Step 3: discrete multiplicities and a uniform Riemann estimate.}
The difference roots have height index \(m=j-i\).  Their multiplicity is
exactly
\begin{equation}
 N_-(m)=\#\{(i,j):1\leq i<j\leq n,\ j-i=m\}=n-m,
 \qquad 1\leq m\leq n-1.
 \label{S:difference-multiplicity}
\end{equation}
Hence, for a test function \(f\),
\begin{equation}
 \frac1{n^2}\sum_{i<j}f\left(\frac{j-i}{n}\right)
 =\frac1n\sum_{m=1}^{n-1}\left(1-\frac mn\right)f(m/n)
 \longrightarrow\int_0^1w_-(z)f(z)\dd z,
 \quad w_-(z)=1-z.
 \label{S:difference-density}
\end{equation}

For the sum roots, first count pairs with \(i+j=u\):
\begin{equation}
 N_\Sigma(u)=
 \begin{cases}
  \lfloor(u-1)/2\rfloor,&3\leq u\leq n+1,\\
  \lfloor(u-1)/2\rfloor-u+n+1,&n+1<u\leq2n-1.
 \end{cases}
 \label{S:sum-multiplicity}
\end{equation}
Writing \(\tau=u/n\), this is
\(n\min(\tau,2-\tau)/2+O(1)\).  The root variable used below is
\(z=2-\tau\), up to the bounded shifts displayed in the root-height table
above.  Therefore
\begin{equation}
 \frac1{n^2}\sum_{i<j}f\left(2-\frac{i+j}{n}\right)
 \longrightarrow\int_0^2w_+(z)f(z)\dd z,
 \qquad
 w_+(z)=
 \begin{cases}
  z/2,&0<z<1,\\
  (2-z)/2,&1<z<2.
 \end{cases}
 \label{S:sum-density}
\end{equation}
Equations~\eqref{S:difference-density} and \eqref{S:sum-density} are the
claimed difference- and sum-root multiplicity densities; they are derived
from the finite multiplicities rather than postulated.

The kernels are logarithmically singular at the zero-height face, so we also
record the error estimate used in passing from these sums to their integrals.

\begin{lemma}[Uniform logarithmic Riemann sums]
\label{lem:uniform-log-Riemann}
Let \(Q\) be positive and smooth on \(I\), and let \(U<\inf_{t\in I}Q(t)\).
For
\[
 \Phi_t(u)=-\log\left(2\sin\frac{\pi u}{Q(t)}\right),
 \qquad 0<u\leq U,
\]
either triangular pair sum in Eqs.~\eqref{S:difference-density} and
\eqref{S:sum-density} equals \(n^2\) times its density integral plus
\(O(n\log n)\).  The same conclusion holds after bounded integer or
half-integer shifts of the root heights and bounded shifts of the denominator
\(nQ(t)\).
\end{lemma}

\begin{proof}
Uniformly for \(t\in I\), split
\begin{equation}
 \Phi_t(u)
 =-\log u+g_t(u).
\end{equation}
Because the support stays a uniform positive distance from the next sine
zero, \(g_t\) extends to \(u=0\) with a uniformly bounded first derivative.
Ordinary cell comparison therefore gives an \(O(n)\) error for the smooth
part.  For \(-\log u\), group points into the \(O(n)\) strips of
Eqs.~\eqref{S:difference-multiplicity} and \eqref{S:sum-multiplicity}.  The
strip multiplicity is \(nw(u)+O(1)\), with \(w\) piecewise affine.  On the
\(m\)-th strip the variation of \(-\log u\) across a cell is \(O(1/m)\).
Thus the cell-comparison error is bounded by a constant times
\(n\sum_{m\leq O(n)}m^{-1}=O(n\log n)\).  The \(O(1)\) discrepancy in each
strip multiplicity contributes only \(O(n)\), since
\begin{equation}
 \sum_{j=1}^{O(n)}\left|\log\frac{j+\delta}{n}\right|=O(n),
 \qquad
 \sum_{j=1}^{O(n)}\frac1{j+\delta}=O(\log n).
\end{equation}
A bounded shift of a root height changes the kernel on the \(m\)-th strip by
\(O(1/m)\); multiplying by the \(O(n)\) points in that strip and summing again
gives \(O(n\log n)\).  A bounded shift of the denominator changes the scaled
denominator by \(O(n^{-1})\) and contributes \(O(n)\) over \(O(n^2)\)
points.  Integer and half-integer grids obey the same estimates.  This proves
the lemma.
\end{proof}

In our applications, type \(A\) has support \(0\leq u\leq1\) and
\(Q(t)=1+t\), whereas types \(B,C,D\) have support \(0\leq u\leq2\) and
\(Q(t)=2(1+t)\).  Since \(t\geq a>0\), the hypothesis
\(U<\inf_IQ(t)\) holds uniformly; the only sine zero approached by the grids
is the logarithmic zero at \(u=0\).

\paragraph{Step 4: evaluation of the four root sums.}
For \(A_{n-1}\), Eq.~\eqref{S:finite-A} and the exact multiplicity
Eq.~\eqref{S:difference-multiplicity} give
\begin{align}
 \log\cD_A(n,k)
 &=-\sum_{m=1}^{n-1}(n-m)
 \log\left(2\sin\frac{\pi m}{n+k}\right)+O(n\log n)\nonumber\\
 &=n^2\int_0^1(1-z)
 \left[-\log\left(2\sin\frac{\pi z}{1+t}\right)\right]\dd z
 +O(n\log n)\nonumber\\
 &=n^2J(t)+O(n\log n),\qquad t=k/n.
 \label{S:A-continuum}
\end{align}

For the other three families define the common limiting kernel
\begin{equation}
 \Phi_t(z)=-\log\left(2\sin\frac{\pi z}{2(1+t)}\right),
 \qquad 0<z<2.
 \label{S:BCD-kernel}
\end{equation}
For \(C_n\), the two scaled height variables in Eq.~\eqref{S:finite-C} are
\((j-i)/n\) and \((2n+2-i-j)/n\), and
\(2K/n=2(1+t)+2/n\).  Lemma~\ref{lem:uniform-log-Riemann} removes the bounded
numerator and denominator shifts.  For \(B_n\) they are
\((j-i)/n\) and \((2n+1-i-j)/n\), with
\(K/n=2(1+t)-1/n\); for \(D_n\) they are
\((j-i)/n\) and \((2n-i-j)/n\), with
\(K/n=2(1+t)-2/n\).  Thus all three families have the same pair-root limit.

Extending \(w_-\) by zero on \((1,2)\), the discrete calculations above give
\begin{equation}
 w_-(z)+w_+(z)=\frac{2-z}{2},\qquad 0<z<2.
 \label{S:combined-density}
\end{equation}
Consequently, with the substitution \(z=2x\),
\begin{align}
 \int_0^2\bigl[w_-(z)+w_+(z)\bigr]\Phi_t(z)\dd z
 &=\frac12\int_0^2(2-z)\Phi_t(z)\dd z\nonumber\\
 &=2\int_0^1(1-x)
 \left[-\log\left(2\sin\frac{\pi x}{1+t}\right)\right]\dd x
 =2J(t).
 \label{S:two-channel-integral}
\end{align}
Adding the already bounded prefactors and one-body roots proves
\begin{equation}
 \log\cD_X(n,k)=2n^2J(t)+O(n\log n),\qquad
 X=B,C,D,
 \label{S:BCD-continuum}
\end{equation}
with \(t=k/n\) for \(C\) and \(t=k/(2n)\) for \(B,D\).

Equations~\eqref{S:A-continuum} and \eqref{S:BCD-continuum} also make the
family factor transparent.  Type \(A\) has only the difference-root channel
and gives \(\mu_A=1\).  Types \(B,C,D\) have both difference and sum channels;
their combined density in Eq.~\eqref{S:combined-density} gives
\(\mu_B=\mu_C=\mu_D=2\).  Together with the rank calculation in Step 1, this
completes the step-by-step proof of Eq.~\eqref{S:asym}, including uniformity
for \(t\) in every compact subset of \((0,\infty)\).

\end{document}